\documentclass[12pt]{amsart}

\usepackage{amsmath,amsfonts,amsbsy}
\usepackage{eucal}
\usepackage{dsfont}
\usepackage{MnSymbol}

\usepackage{xy}
\xyoption{matrix} \xyoption{arrow} \xyoption{arc} \xyoption{color}

\UseCrayolaColors
\usepackage{enumerate}
\usepackage{enumitem}

\usepackage{subfigure}
\usepackage{tikz}
\usetikzlibrary{arrows}

\makeatletter
\newtheoremstyle{corsivo}
   {\medskipamount}{\medskipamount}%
   {\itshape}{}%
   {\bfseries}{}%
   { }
   {\thmname{#1}\thmnumber{\@ifnotempty{#1}{ }\@upn{#2}}%
    \thmnote{ {\bfseries(#3)}}.}%
\makeatother

\theoremstyle{corsivo}
\newtheorem{thm}{Theorem}[section]
\newtheorem{lemma}[thm]{Lemma}
\newtheorem{crl}[thm]{Corollary}
\newtheorem{prop}[thm]{Proposition}

\numberwithin{equation}{section}

\makeatletter
\newtheoremstyle{dritto}
   {\medskipamount}{\medskipamount}%
   {\rmfamily}{}%
   {\bfseries}{}%
   { }
   {\thmname{#1}\thmnumber{\@ifnotempty{#1}{ }\@upn{#2}}%
    \thmnote{ {\bfseries(#3)}}.}%
\makeatother

\theoremstyle{dritto}
\newtheorem{rmk}[thm]{Remark}

\newtheorem{dfn}[thm]{Definition}
\newtheorem{assumption}[thm]{Assumption}

\newcommand{\Id}{\mathds{1}}  
\newcommand{\id}{\mathbb{I}}   
\newcommand{\B}{\mathbb{B}}
\newcommand{\Bred}{\mathbb{B}\sub{eff}}  

\newcommand{\C}{\mathbb{C}}
\newcommand{\R}{\mathbb{R}}
\newcommand{\Z}{\mathbb{Z}}
\newcommand{\N}{\mathbb{N}}
\newcommand{\T}{\mathbb{T}}

\newcommand{\PB}{\mathcal{P}}
\newcommand{\Hi}{\mathcal{H}}
\newcommand{\U}{\mathcal{U}}

\newcommand{\BH}{\mathcal{B}(\mathcal{H})}
\newcommand{\scal}[2]{\left\langle #1, #2 \right\rangle}

\newcommand{\bra}[1]{\left\langle #1 \right|}
\newcommand{\ket}[1]{\left| #1 \right\rangle}
\newcommand{\eu}{\mathrm{e}}
\newcommand{\iu}{\mathrm{i}}
\newcommand{\di}{\mathrm{d}}
\newcommand{\act}{\triangleleft}  
\newcommand{\sub}[1]{_{\mathrm{#1}}}
\newcommand{\Tr}{^{\mathsf{T}}}  

\newcommand{\inner}[2]{\left\langle #1, #2 \right\rangle}
\newcommand{\norm}[1]{\left\| #1 \right\|}
\newcommand{\virg}[1]{``#1\,''}
\newcommand{\half}{\mbox{\footnotesize $\frac{1}{2}$}}

\newcommand{\cell}[1]{{#1}_{\rm eff}}
\newcommand{\uno}{_{\rm one}}
\newcommand{\due}{_{\rm two}}
\newcommand{\tre}{_{\rm three}}
\newcommand{\quat}{_{\rm four}}
\newcommand{\cinque}{_{\rm five}}

\newcommand{\la}{\lambda}  
\newcommand{\La}{\Lambda}  
\newcommand{\epsi}{\varepsilon}

\newcommand{\UP}{U_{\Phi, \Psi}}
\newcommand{\midpoint}[2]{\wideparen{M} \left(  #1, #2 \right)}   

\def\({\left(}
\def\){\right)}

\DeclareMathOperator{\tr}{tr}
\DeclareMathOperator{\Ran}{Ran}
\DeclareMathOperator{\Fr}{Fr}
\DeclareMathOperator{\diag}{diag}

\DeclareMathOperator{\dist}{dist}

\newcommand{\ie}{{\sl i.\,e.\ }}
\newcommand{\eg}{{\sl e.\,g.\ }}

\newcommand{\set}[1]{ \left\{  #1 \right\}}

\let\oldfootnote\footnote
\renewcommand{\footnote}[1]{\oldfootnote{\  #1}}

\title[Construction of real localized Wannier functions]{Construction of  
real-valued localized \\[3mm] composite Wannier functions for insulators}
\author[D. Fiorenza, D. Monaco, G. Panati]{Domenico Fiorenza, Domenico Monaco, Gianluca Panati}

\usepackage{hyperref}

\usepackage{color}

\begin{document}

\begin{abstract} We consider a real periodic Schr\"{o}dinger operator and a 
physically relevant family of  $m \geq 1$  Bloch bands, 
separated by a gap from the rest of the spectrum, and we
 investigate the localization properties of the corresponding 
composite Wannier functions. To this aim, we  show that in dimension
$d \leq 3$ there exists a global frame consisting of smooth quasi-Bloch  
functions which are both periodic and time-reversal symmetric.  
Aiming to applications in computational physics, we provide a constructive algorithm to obtain such a Bloch frame. The construction yields the existence of a basis of composite Wannier functions which are real-valued and almost-exponentially localized. 

The proof of the main result exploits only the fundamental symmetries of the projector on the relevant bands, 
allowing applications, beyond the model specified above, to a broad range of gapped periodic quantum systems with a time-reversal symmetry of bosonic type.

\smallskip

\vspace{-9mm}

\end{abstract}

\maketitle

\tableofcontents

\newpage
\section{Introduction}
\label{Sec:Introduction}

The existence of an orthonormal basis of well-localized  Wannier functions  is a crucial 
issue in solid-state physics \cite{Wannier review}.  Indeed, such a basis is the key tool to obtain effective tight-binding models for a linear or non-linear Schr\"odinger dynamics 
\cite{PelinovskySchneiderMacKay2008, PelinovskySchneider2010, 
Ibanez-Azpiroz_et_al2013a, Ibanez-Azpiroz_et_al2013b, Walters_et_al2013}, it allows computational methods whose cost scales linearly with the size of the confining box \cite{Goedecker},  it is useful in the rigorous analysis of perturbed periodic Hamiltonians 
\cite{CancesLewin2010, LewinSere2010, ELu2011}, and it  is crucial in the modern theory of polarization of crystalline solids  \cite{Resta92, KSV, PanatiSparberTeufel2009} and in the pioneering research on topological insulators \cite{Hastings2008,  RyuSchnyder2010, Prodan2009, Prodan2011b, SoluyanovVanderbilt2011, SoluyanovVanderbilt2011b}.


In the case of a single isolated Bloch band, which does not touch any other Bloch band, the rigorous 
proof of the existence of exponentially localized Wannier functions goes back to the work of W. Kohn 
 \cite{Kohn59}, who provided a proof in dimension $d =1$ for an even potential. The latter assumption was later removed 
 by J. de Cloizeaux  \cite{Cl2},  who also gave a proof valid for any $d > 1$ under the assumption that the periodic 
 potential is centro-symmetric \cite{Cl1, Cl2}.    
The first proof under  generic assumptions, again for any $d >1$, was provided by G. Nenciu \cite{Ne83}, 
and few years later a simpler proof appeared  \cite{HeSj89}.  

In real solids, Bloch bands intersect each other. Therefore, as early suggested \cite{Bl, Cl2}, it is more natural to focus on 
a family of  $m$ Bloch bands which is separated by a gap from the rest of the spectrum, as \eg the family of all the Bloch bands below the Fermi energy in an insulator or a semiconductor. Accordingly, the notion of Bloch function is weakened to that of \emph{quasi-Bloch function} and, correspondingly, one considers \emph{composite Wannier functions} (Definition \ref{Def:CompositeWannierBasis}).  In the multi-band case, the existence of exponentially localized composite Wannier functions is subtle, since it might be topologically obstructed. A proof of existence was provided in \cite{NeNe82, Ne91} for $d =1$, while 
a proof in the case $d \leq 3$ required more abstract bundle-theoretic methods \cite{BrouderPanati2007, Panati2007}, both results being valid for any number of bands $m \in \N$. In the $1$-dimensional case generalizations to non-periodic gapped systems are also possible \cite{NeNe97}, as well as extensions to quasi-$1$-dimensional systems \cite{CoNe}.

 \medskip
 
Beyond the abstract existence results, computational physics strived for an explicit construction.
On the one hand, Marzari and Vanderbilt \cite{MaVa} suggested a shift to a variational viewpoint,  
which is nowadays very popular in computational solid-state physics.  They introduced a suitable 
localization functional, defined on a set of composite Wannier functions, and argued that the 
corresponding  minimizers are expected to be exponentially localized. They also noticed that, for $d=1$, 
the minimizers are indeed exponentially localized in view of the relation between the composite Wannier functions 
and the eigenfunctions of the reduced position operator  \cite{Kievelsen, NeNe97}. 
For $d >1 $, the exponential localization of the minimizers follows instead from deeper properties of the localization 
functional \cite{PanatiPisante}, if $d \leq 3$ and some technical hypotheses are satisfied. 
Moreover, there is numerical evidence that the minimizers are \emph{real-valued}  functions, 
but a mathematical proof of this fact is still missing \cite{MaVa, BrouderPanati2007}.  

On the other hand, researchers are also working to obtain an explicit algorithm yielding composite Wannier 
functions which are both real-valued and well-localized \cite{CoNe2014}. As a predecessor in this direction, we  
mention again the result in \cite{HeSj89}, which provides an explicit proof in the single-band case, \ie for $m=1$, through the construction of time-reversal symmetric Bloch functions (see below for detailed comments).    

In this paper, following the second route, we provide an explicitly constructive algorithm to obtain, for any $d \leq 3$ and $m \in \N$, composite Wannier functions which are {\emph{real-valued}} and \emph{almost-exponentially localized},  in the sense that they decay faster than the inverse of any polynomial (Theorem \ref{Thm:RealLocalizedWannierBasis}). The latter result follows from a more general theorem (Theorem \ref{Thm:SmoothBlochFrames}), which applies to a broad range of gapped periodic quantum systems with a time-reversal symmetry of \emph{bosonic} (or \emph{even}) type (see Assumption \ref{Ass:projectors}).  
Under such an assumption, we explicitly construct a smooth frame of eigenfunctions of the relevant projector (\ie quasi-Bloch functions in the application to Schr\"{o}dinger operators) which are both pseudo-periodic and time-reversal symmetric, in the sense of Definition \ref{dfn:Bloch}. Since the result is proved in a general setting, we foresee possible applications  to periodic Pauli or Dirac operators, as well as to tight-binding models as \eg the one proposed by Haldane \cite{Haldane88}.   
Despite the apparent similarity, the case of systems with  \emph{fermionic} (or \emph{odd}) time-reversal symmetry, 
relevant in the context of topological insulators \cite{HasanKane},  
is radically different, as emphasized in  \cite{SoluyanovVanderbilt2012, GrafPorta, FiMoPa}, see Remark \ref{Rem:FermionicSystems}.


We conclude the Introduction with few comments about the relation between our constructive algorithm and  the proofs of some previous results. 
 
\noindent The proof provided by Helffer and Sj\"ostrand for $m=1$  and $d \in \N$ \cite{HeSj89}, 
is explicitly constructive and yields real-valued Wannier functions. However, the proof has not a direct
generalization to the case $m>1$ for a very subtle reason, which is occasionally overlooked even by experts. 
We illustrate the crucial difficulty in the simplest case, by considering a unitary matrix $U(k_1) \in \U(\C^m)$ 
depending continuously on a parameter $k_1 \in \T^1 = \R/2\pi\Z$.  When mimicking the proof in   
\cite{HeSj89}, one defines (\eg via spectral calculus) the unitary $U(k_1)^{k_2}$, for $k_2 \in [0,1/2]$,  which is   
well-defined whenever a determination of the complex logarithm has been chosen in such a way that the 
branch-cut does not touch the (point) spectrum of  $U(k_1)$.  As $k_1 \in \T^1$ varies, the branch-cut must vary accordingly, 
and it might happen that the branch-cut for $k_1 = 2 \pi$ equals the one for $k_1 =0$ after a complete wind (or more) 
in the complex plane. In such an eventuality, the rest of the argument fails.  
In \cite{Ne83}, a similar difficulty appears.\footnote{We cite textually from \cite{Ne83}: \emph{Unfortunately, we have 
been unable to prove that $T(\mathbf{z}^{q-1})$ admits an analytic and periodic logarithm [\ldots], and therefore we shall 
follow a slightly different route.}}   As far as we know,  there is no direct way to circumvent this kind of difficulty. For this reason, in this paper we develop a radically different technique.

\medskip \goodbreak

The paper is organized as follows.  In Section \ref{Sec:Schrodinger} we consider a real periodic Schr\"odinger operator and we show that, for a gapped system as \eg an insulator, the orthogonal projector on the Bloch states up to the gap  satisfies some natural properties (Proposition \ref{Prop P properties}).  Generalizing from the specific example, the abstract version of these properties becomes our starting point, namely Assumption \ref{Ass:projectors}. In Section \ref{Sec:results} we state 
our main results, and we briefly comment on the structure of the proof, which is the content of Section \ref{Sec:Algorithm}. Finally, a technical result concerning the smoothing of a continuous symmetric Bloch frame to obtain a smooth symmetric Bloch frame, which holds true in any dimension and might be of independent interest, is provided in Section \ref{Sec:SmoothingProcedure}. 


\bigskip

\textbf{Acknowledgments.}  We are indebted to A. Pisante for many useful comments. G.P. is  grateful to H. Cornean and G. Nenciu for useful discussions,  and to H.~Spohn and S. Teufel for stimulating his interest in this problem during the preparation of \cite{PST2003}.  We are  grateful to the {\it  Institut Henri Poincar\'e} for the kind hospitality in the framework of the  trimester \virg{Variational and Spectral Methods in Quantum Mechanics}, organized by M.\,J.\ Esteban and M.\ Lewin. 

\noindent This project was supported by the National Group for Mathematical Physics (INdAM-GNFM)  and from MIUR (Project PRIN 2012).

\newpage

\section{From Schr\"{o}dinger operators to covariant families of projectors}
\label{Sec:Schrodinger}

\newcommand{\UZ}{\U\sub{BF}}
\newcommand{\Hf}{\Hi\sub{f}}

The dynamics of a particle in a crystalline solid can be modeled by use of a \emph{periodic Schr\"{o}dinger operator}
\[ H_{\Gamma} = - \Delta + V_{\Gamma} \qquad \text{acting in } L^2(\R^d), \]
where the potential $V_\Gamma$ is periodic with respect to a lattice (called the \emph{Bravais lattice} in the physics literature)
\[ \Gamma := \mathrm{Span}_\Z\set{a_1, \ldots, a_d}\simeq \Z^d \subset \R^d, \qquad \text{ with }\set{a_1, \ldots, a_d} \text{ a basis in } \R^d. \]
Assuming that
\begin{equation} \label{Katosmallness}
V_{\Gamma} \in L_{\rm loc}^2(\R^d) \mbox{ for } d \leq 3, \qquad \text{or} \qquad V_{\Gamma} \in L_{\rm loc}^p(\R^d) \mbox{ with } p > d/2 \mbox{ for } d \geq 4, 
\end{equation}
the operator $H_\Gamma$ is self-adjoint on the domain $H^2(\R^d)$ \cite[Theorem XIII.96]{RS4}.

In order to simplify the analysis of such operators, one looks for a convenient representation which (partially) diagonalizes simultaneously both the Hamiltonian and the lattice translations. This is provided by the (\emph{modified}) \emph{Bloch-Floquet transform}, defined on suitable functions $w \in  C_0(\R^d) \subset L^2(\R^d)$ as
\begin{equation} \label{Zak transform}
( \UZ \, w)(k,y):= \frac{1}{|\B|^{1/2}} \,\, \sum_{\gamma\in\Gamma} \eu^{-\iu k \cdot (y +\gamma)} \, w( y + \gamma), \qquad y \in \R^d, \, k \in\R^{d}.
\end{equation}
Here $\B$ is the fundamental unit cell for the dual lattice $\Gamma^* := \mathrm{Span}_\Z\set{b_1, \ldots, b_d} \subset \R^d$, determined by the basis $\set{b_1, \ldots, b_d}$ which satisfies $b_i \cdot a_j = 2 \pi \delta_{ij}$, namely
\[ \B := \set{k = \sum_{j=1}^{d} k_j b_j \in \R^d: -\half  \le k_j \le \half }.  \]
From \eqref{Zak transform}, one immediately reads the (pseudo-)periodicity properties
\begin{equation}\label{Zak properties}
\begin{aligned}
\big( \UZ \, w\big) (k, y+\gamma) &= \big( \UZ\, w\big)(k,y) && \mbox{for all } \gamma \in\Gamma\,,\\
\big( \UZ \, w\big) (k + \lambda, y) &= \eu^{-\iu \la \cdot y}\,\big( \UZ \, w \big) (k,y) && \mbox{for all }\lambda\in\Gamma^*\,.
\end{aligned}
\end{equation}
The function $\big( \UZ \, w\big) (k, \cdot)$, for fixed $k \in \R^d$, is thus periodic, so it can be interpreted as an element in the Hilbert space $\Hf := L^2(\T^d_Y)$, where $\T^d_Y = \R^d / \Gamma$ is the torus obtained by identifying opposite faces of the fundamental unit cell for $\Gamma$, given by
\[ Y := \set{y = \sum_{j=1}^{d} y_j a_j \in \R^d: -\half \le y_j \le \half }. \]

Following \cite{PST2003}, we reinterpret \eqref{Zak properties} in order to emphasize the role of covariance with respect to the action of the relevant symmetry group. Setting
\[ \big(\tau(\lambda)\psi \big)(y) := \eu^{- \iu \lambda \cdot \,y} \psi(y), \qquad \text{for } \psi \in \Hf, \]
one obtains a unitary representation $\tau \,\colon\, \Gamma^* \to \U(\Hf)$ of the group of translations by vectors of the dual lattice.
One can then argue that $\UZ$ establishes a unitary transformation $\UZ : L^2(\R^d) \to \Hi_\tau$, where $\Hi_\tau$ is the Hilbert space
$$
\Hi_\tau :=\Big\{ \phi \in L^2_{\rm loc}(\R^d, \Hf ):\,\, \phi(k + \lambda) = \tau(\lambda)\,\phi(k) \; \forall \lambda \in \Gamma^{*}, \mbox{ for a.e. } k \in \R^d \Big\}\, 
$$
equipped with the inner product 
$$
\inner{\phi}{\psi}_{\Hi_{\tau}} =  \int_{\B} \inner{\phi(k)}{\psi(k)}_{\Hf} \, \di k.
$$
Moreover, the inverse transformation $\UZ^{-1} : \Hi_\tau \to L^2(\R^d)$ is explicitely given by
\[ \left( \UZ^{-1} \phi \right)(x) =  \frac{1}{|\B|^{1/2}} \int_{\B} \di k \,\eu^{ \iu k \cdot x} \phi(k, x). \]

In view of the identification
\[ \Hi_\tau \simeq \int_{\B}^{\oplus} \di k \: \Hf, \]
we see that  the Schr\"{o}dinger operator $H_\Gamma$ becomes a fibered operator in the Bloch-Floquet representation, \ie
\[ \UZ \, H_{\Gamma} \, \UZ^{-1} = \int_{\B}^\oplus \di k \,H(k), \qquad \text{where} \qquad H(k) = \big( -\iu \nabla_y + k\big)^2 + V_\Gamma(y). \]
The fiber operator $H(k)$, $k \in \R^d$, acts on the $k$-independent domain $H^2(\T_Y^d) \subset \Hf$, where it defines a self-adjoint operator. Moreover, it has compact resolvent, and thus pure point spectrum. We label its eigenvalues, accumulating at infinity, in increasing order, as $E_0(k) \le E_1(k) \le \cdots \le E_n(k) \le E_{n+1}(k) \le \cdots$, counting multiplicities. The functions $\R^d \ni k \mapsto E_n(k) \in \R$ are called \emph{Bloch bands}. Since the fiber operator $H(k)$ is \emph{$\tau$-covariant}, in the sense that 
\[ H(k + \lambda) = \tau(\lambda)^{-1} \, H(k) \, \tau(\lambda), \qquad \lambda \in \Gamma^*, \]
Bloch bands are actually periodic functions of $k \in \R^d$, \ie $E_n(k+\lambda) = E_n(k)$ for all $\lambda \in \Gamma^*$, and hence are determined by the values attained at points $k \in \B$.

A solution $u_n(k)$ to the eigenvalue problem
\[ H(k) u_n(k) = E_n(k) u_n(k), \qquad u_n(k) \in \Hf, \qquad \norm{u_n(k)}_{\Hf} = 1, \]
constitutes the (periodic part of the) \emph{$n$-th Bloch function}, in the physics terminology. Assuming that, for fixed 
$n \in \N$, the eigenvalue $E_n(k)$ is non-degenerate for all $k \in \R^d$, the function $u_n: y \mapsto u_n(k,y)$ is 
determined up to the choice of a $k$-dependent phase, called the \emph{Bloch gauge}.

\goodbreak

By definition, the \textbf{Wannier function}  $w_n$ corresponding to the Bloch function $u_n \in \Hi_{\tau}$ is the preimage, via Bloch-Floquet transform, of the Bloch function, \ie
\begin{equation} \label{Wannier}
w_n(x) := \left( \UZ^{-1} u_n \right)(x) = \frac{1}{|\B|^{1/2}} \int_{\B} \di k \,\eu^{ \iu k \cdot x} u_n(k, x).
\end{equation}
Localization (\ie decay at infinity) of the Wannier function $w = w_n$ and smoothness of the associated Bloch function $u = u_n$ are related by the following statement, that can be checked easily from the definition \eqref{Zak transform} of the Bloch-Floquet transform (see \cite[Sec. 2]{PanatiPisante} for details):
\begin{equation} \label{Zak equivalence}
\begin{aligned}
w \in H^s(\R^d), \ s \in \N & \Longleftrightarrow u \in L^2(\B, H^s(\T_Y^d)), \\
\langle x \rangle^r w \in L^2(\R^d), \ r \in \N & \Longleftrightarrow u \in \Hi_\tau \cap H^r_{\rm loc}(\R^d, \Hf),\\
\end{aligned}
\end{equation}
where we used the Japanese bracket notation $\langle x \rangle = (1 + |x^2|)^{1/2}$. A Wannier function such that $\langle x \rangle^r w \in L^2(\R^d)$ for all $r \in \N$ will be called \textbf{almost-exponentially localized}.


As mentioned in the Introduction, to deal with real solids, where generically the Bloch bands intersect each other,  a multi-band theory becomes necessary.   Many of the above statements can be formulated even in the case when more than one Bloch band is considered. Let $\sigma_*(k)$ be the set $\set{E_{i}(k): n \leq i \leq n+m-1}$, $k \in \B$, corresponding to a family of $m$ Bloch bands. Usually, in the applications,  $\sigma_*(k)$ consists of some Bloch bands which are physically relevant, as \eg the bands below the Fermi energy in insulators and semiconductors. Assume the following \emph{gap condition}:
\begin{equation}\label{Gap condition}
\inf_{k \in \B} \mathrm{dist}\big( \sigma_*(k), \sigma(H(k)) \setminus \sigma_*(k) \big) > 0.
\end{equation}
The relevant object to consider in this case is then the \emph{spectral projector} $P_*(k)$ on the set $\sigma_*(k)$, which in the physics literature  reads 
$$
P_*(k) = \sum_{n \in \mathcal{I}_*} \ket{u_n(k)}\bra{u_n(k)}, 
$$
where the sum runs over all the bands in the relevant family, \ie over the set $ \mathcal{I}_* = \set{n \in \N :E_n(k) \in \sigma_*(k)}$.
As proved in \cite[Prop.\ 2.1]{PanatiPisante},  elaborating on a longstanding tradition of related results \cite{RS4, Ne91},  
the projector $P_*(k)$   satisfies the properties listed in the following Proposition.

\begin{prop} \label{Prop P properties}
Let $P_*(k) \in \mathcal{B}(\Hf)$ be the spectral projector of $H(k)$ corresponding to the set $\sigma_*(k) \subset \R$. Assume that $\sigma_{*}$ satisfies the gap condition \eqref{Gap condition}. Then the family $\set{P_*(k)}_{k \in \R^d}$ has the following properties:
\begin{enumerate}[label=$(\mathrm{p}_{\arabic*})$, ref=$(\mathrm{p}_{\arabic*})$]
\item the map $k \mapsto P_*(k)$ is smooth from $\R^d$ to $\mathcal{B}(\Hf)$ (equipped with the operator norm);
\item the map $k \mapsto P_*(k)$ is $\tau$-covariant, \ie
\[ P_*(k + \lambda) = \tau(\lambda) \, P_*(k) \, \tau(\lambda)^{-1}  \qquad \forall k \in \R^d, \quad \forall \lambda \in \Gamma^{*}; \]
\item \label{p3} there exists an antiunitary operator\footnote{By \emph{antiunitary} operator we mean a surjective antilinear operator $C: \Hi \rightarrow \Hi$,    such that $\scal{C\phi}{C \psi}_{\Hi} = \scal{\psi}{\phi}_{\Hi}$ for any $\phi, \psi \in \Hi$.} 
$C$ acting on $\Hf$ such that
\[ P_*(-k) =  C \, P_*(k) \, C^{-1}  \qquad \mbox{ and  } \qquad C^{2} =1. \]
\end{enumerate}
\end{prop}

The antiunitary operator $C$ appearing in \ref{p3} is explicitly given by the complex conjugation in $\Hf = L^2(\T^d_Y)$ and, in particular,  one has $C \tau(\la) = \tau(-\la) C$ for all $\la \in \Gamma^*$. 

In the multi-band case, it is convenient  \cite{Bl, Cl1}  to relax the notion of Bloch function and to consider \emph{quasi-Bloch functions}, defined as elements $\phi \in \Hi_{\tau}$ such that
\[ P_*(k) \phi(k) = \phi(k), \qquad \norm{\phi(k)}_{\Hf}=1, \qquad \text{ for a.e. }  k \in \B. \]
A \emph{Bloch frame} is, by definition,  a family of quasi-Bloch functions $\set{\phi_a}_{a=1, \ldots, m}$, constituting an orthonormal basis of $\Ran P_*(k)$ at a.e.\ $k \in \B$. \newline
In this context, a non-abelian Bloch gauge appears, since whenever $\set{\phi_a}$ is a Bloch frame, then one obtains another Bloch frame $\set{\widetilde \phi_a}$ by setting
$$
\widetilde \phi_a (k) = \sum_{b =1}^m \phi_b(k) \, U_{ba} (k)  \qquad \text{ for some unitary matrix } U(k).  
$$
Equipped with this terminology, we rephrase a classical definition \cite{Cl2} as follows:  

\begin{dfn}[Composite Wannier functions] \label{Def:CompositeWannierBasis}
The \emph{composite Wannier functions}  $\set{w_1, \ldots, w_m} \subset L^2(\R^d)$  associated to a Bloch frame 
$\set{\phi_1, \ldots, \phi_m} \subset \Hi_{\tau}$ are defined as 
\[ w_a(x) := \left( \UZ^{-1} \phi_a \right)(x) = \frac{1}{|\B|^{1/2}} \int_{\B} \di k \,\eu^{ \iu k \cdot x} \phi_a(k, x). \]
\end{dfn}

An orthonormal basis of $ \UZ^{-1}  \Ran P_*$ is readily obtained by considering the translated functions
$$
w_{\gamma, a} (x) :=   w_{a}(x - \gamma).
$$
In view of the orthogonality of the trigonometric polynomials, the set $\set{w_{\gamma, a} }_{\gamma \in \Gamma, 1 \leq a \leq m}$ is an orthonormal basis of $ \UZ^{-1}  \Ran P_*$, which we refer to as  a \emph{composite Wannier basis}. 
The above-mentioned Bloch gauge freedom implies that the latter basis is not unique, and its properties (\eg localization) will in general depend on the choice of a Bloch gauge.  

As emphasized in the Introduction, the existence of an orthonormal basis of well-localized Wannier functions is a crucial issue in solid-state physics. In view of \eqref{Zak equivalence}, the existence of a composite Wannier basis consisting of almost-exponentially localized functions is equivalent to the existence of a $C^\infty$-smooth Bloch frame for $\set{P_*(k)}_{k \in \R^d}$. The existence of the latter might be \emph{a priori} obstructed since, as noticed by several authors \cite{Kohn59, Cl1, Ne91}, 
there might be  competition between the smoothness of the function $k \mapsto \phi_a(k)$ and its pseudo-periodicity properties, here encoded in the fact that $\phi_a \in \Hi_{\tau}$ must satisfy \eqref{Zak properties}.  
\emph{A posteriori},  it has been proved that,  as a consequence of the time-reversal symmetry of the system, encoded in property \ref{p3}, this obstruction is absent, yielding the existence of a $C^{\infty}$-smooth (actually, analytic) Bloch frame for any $d \leq 3$ and $m \in \N$  \cite{Panati2007, BrouderPanati2007}.  The  result in \cite{Panati2007}, however, neither provides explicitly such a Bloch frame, nor it guarantees that it is time-reversal symmetric. In the next Sections, we tackle these problems in a more general framework.

\newpage
\section{Assumptions and main results}\label{Sec:results}

Abstracting from the case of periodic Schr\"{o}dinger operators, we state our results in a general setting. 
Our assumptions are designed to rely only on two fundamental symmetries of the system, 
namely covariance with respect to translations by vectors in the dual lattice  
and a time-reversal  symmetry of bosonic type, \ie with a time-reversal operator $\Theta$ satisfying $\Theta^2 = \Id$  (see Remark \ref{Rem:FermionicSystems} for the fermionic case).  In view of that, the following abstract results apply both to continuous models, as \eg the real Schr\"{o}dinger operators considered in the previous Section, and to discrete models, as \eg 
 the Haldane model \cite{Haldane88}.   

In the following, we let $\Hi$ be a separable Hilbert space with scalar product $\scal{\cdot}{\cdot}$, $\BH$ denote the algebra of  bounded linear operators on $\Hi$, and $\U(\Hi)$ the group of unitary operators on $\Hi$. We also consider a maximal lattice $\Lambda \simeq \Z^d \subset \R^d$ which, in the application to Schr\"{o}dinger operators, is identified with the dual (or reciprocal) lattice $\Gamma^*$. 

\begin{assumption} \label{Ass:projectors}
We consider a family of orthogonal projectors $\set{P(k)}_{k \in \R^d} \subset \BH$ satisfying the following assumptions:
\begin{enumerate}[label=$(\mathrm{P}_{\arabic*})$,ref=$(\mathrm{P}_{\arabic*})$]
\item \label{item:smooth} \emph{smoothness}: the map $\R^d \ni k \mapsto P(k) \in \BH$ is $C^\infty$-smooth;
\item \label{item:tau} \emph{$\tau$-covariance}: the map $k \mapsto P(k)$ is covariant with respect to a unitary representation%
\footnote{This means that $\tau(0) = \Id_{\Hi}$ and $\tau(\lambda_1 + \lambda_2) = \tau(\lambda_1) \tau(\lambda_2)$ for all $\lambda_1, \lambda_2 \in \Lambda$. It follows in particular that $\tau(\lambda)^{-1} = \tau(\lambda)^* = \tau(-\lambda)$ for all $\lambda \in \Lambda$.} %
$\tau  :   \Lambda \to \U(\Hi)$, $\la \mapsto \tau(\la) \equiv \tau_\la$, in the sense that
\[ P(k+\lambda) = \tau_\la \, P(k) \, \tau_\la^{-1} \quad \text{for all } k \in \R^d, \, \lambda \in \Lambda; \]
\item \label{item:TRS} \emph{time-reversal symmetry}: there exists an antiunitary %
operator $\Theta$ acting on $\Hi$, called the \emph{time-reversal operator}, such that 
\[ P(-k) = \Theta \, P(k) \, \Theta^{-1} \quad \text{and} \quad \Theta^2 =  \Id_{\Hi}. \]
\end{enumerate}
Moreover, we assume the following
\begin{enumerate}[resume,label=$(\mathrm{P}_{\arabic*})$,ref=$(\mathrm{P}_{\arabic*})$]
\item \label{item:TRtau} \emph{compatibility condition}: for all $\lambda \in \Lambda$ one has $\Theta \, \tau_\la = \tau_\la^{-1} \,\Theta$.  \hfill  $\lozenge$ 
\end{enumerate}  
\end{assumption}
 
It follows from the assumption \ref{item:smooth} that the rank $m$ of the projector $P(k)$ is constant in $k$. We will assume  that $m < + \infty$. Proposition \ref{Prop P properties} guarantees that the above assumptions are satisfied by the spectral projectors 
$\set{P_*(k)}_{k \in \R^d}$ corresponding to an isolated family of Bloch bands of a \emph{real} periodic Schr\"{o}dinger operator.  

\bigskip

\begin{dfn}[Symmetric Bloch frame] \label{dfn:Bloch}
Let $\mathcal{P} =\set{P(k)}_{k \in \R^d}$ be a family of projectors satisfying Assumption \ref{Ass:projectors}.
A \textbf{local Bloch frame} for $\mathcal{P}$ on a region $\Omega \subset \R^d$ is a map 
\begin{eqnarray*}
\Phi : \, & \Omega & \longrightarrow \quad \Hi \oplus \ldots \oplus \Hi = \Hi^m \\
& k &\longmapsto  \quad  (\phi_1(k), \ldots, \phi_m(k))
\end{eqnarray*}  such that for a.e. $k \in \Omega$ the set $\set{\phi_1(k), \ldots, \phi_m(k)}$ is an orthonormal basis spanning $\Ran P(k)$.  If $\Omega = \R^d$ we say that $\Phi$ is a \textbf{global Bloch frame}.
Moreover, we say that a (global) Bloch frame is 
\begin{enumerate}[label=$(\mathrm{F}_{\arabic*})$,ref=$(\mathrm{F}_{\arabic*})$]
\setcounter{enumi}{-1}
\item  \emph{continuous} if the map $\phi_a : \R^d  \to \Hi^m$ is continuous for all $ a \in \set{1, \ldots, m}$;
\item  \label{item:F1}  \emph{smooth}  if the map $\phi_a : \R^d  \to \Hi^m$  is $C^{\infty}$-smooth for all $a \in \set{1, \ldots, m}$;
\item   \label{item:F2} \emph{$\tau$-equivariant} if 
\begin{equation*} \label{tau-cov}
\phi_a(k + \lambda) = \tau_\la \, \phi_a(k) \quad \text{for all } k \in \R^d, \: \lambda \in \Lambda, \: a \in \set{1, \ldots, m};
\end{equation*}
\item \label{item:F3} \emph{time-reversal invariant} if
\begin{equation*} \label{tr}
\phi_a(-k) =  \Theta \, \phi_a(k)  \quad \text{for all } k \in \R^d, \: a \in \set{1, \ldots, m}.
\end{equation*}
\end{enumerate}
A global Bloch frame is called {\textbf{symmetric}} if satisfies both $(\mathrm{F}_{2})$ and $(\mathrm{F}_{3})$.
\hfill  $\lozenge$ 
\end{dfn}


\begin{thm}[Abstract result] \label{Thm:SmoothBlochFrames}
Assume $d \leq 3$. Let $\mathcal{P} =\set{P(k)}_{k \in \R^d}$ be a family of orthogonal projectors satisfying Assumption \ref{Ass:projectors},  with finite rank $m \in \N$. Then there exists a global \emph{\textbf{smooth symmetric Bloch frame}} for $\mathcal{P}$. Moreover, the proof is explicitly constructive.
\end{thm}


As mentioned in Section \ref{Sec:Introduction}, the relevance of Theorem \ref{Thm:SmoothBlochFrames} is twofold. On the one hand, it provides the first constructive proof, for $m>1$ and $d>1$, of the existence of smooth $\tau$-equivariant Bloch frames, thus providing an explicit algorithm to obtain an almost-exponentially localized composite Wannier basis. On the other hand, the fact that such a smooth Bloch frame also satisfies \ref{tr} implies the existence of \emph{real-valued} localized composite Wannier functions, a fact indirectly conjectured in the literature about optimally localized Wannier functions, and confirmed by numerical evidence \cite[Section V.B]{MaVa}.  
We summarize these consequences in the following statement. 

\begin{thm}[Application to Schr\"{o}dinger operators] \label{Thm:RealLocalizedWannierBasis}
Assume $d \leq 3$. Consider a real periodic Schr\"{o}dinger operator in the form $H_{\Gamma} = - \Delta + V_{\Gamma} $, with $V_{\Gamma}$ satisfying \eqref{Katosmallness}, acting on $H^2(\R^d) \subset L^2(\R^d)$.    
Let  $\mathcal{P_*} =\set{P_*(k)}_{k \in \R^d}$  be the set of spectral projectors corresponding to a family of $m$ Bloch bands satisfying condition \eqref{Gap condition}. Then one constructs an orthonormal basis $\set{w_{\gamma, a}}_{\gamma \in \Gamma,  1 \leq a \leq  m}$ 
of $\U_{\rm BF}^{-1} \Ran P_{*}$  consisting  of composite Wannier functions such that: 
\begin{enumerate} [label={\rm (\roman*)},ref={\rm (\roman*)}]
\item  \label{item:R-valued} each function $w_{\gamma, a}$ is {\textbf{real-valued}}, and 
\item  \label{item:exp-localized} each function $w_{\gamma, a}$ is  \textbf{almost-exponentially localized}, in the sense that 
$$
\int_{\R^d}  \langle x \rangle^{2r} |w_{\gamma, a}(x)|^2 \di x  < + \infty  \qquad \text{for all } r \in \N. 
$$
\end{enumerate} 
\end{thm}  

\begin{proof}
In view of Proposition \ref{Prop P properties},  the family $\mathcal{P_*} =\set{P_*(k)}_{k \in \R^d}$ satisfies Assumption \ref{Ass:projectors}. Thus, by Theorem \ref{Thm:SmoothBlochFrames}, there exists a global smooth symmetric Bloch frame in the sense of  Definition \ref{dfn:Bloch}. 
In view of \ref{item:F1} and \ref{item:F2},  each $\phi_a$ is an element of $\Hi_{\tau} \cap C^{\infty}(\R^d, \Hf)$, and thus 
$\Phi$ is a smooth Bloch frame in the sense of Section \ref{Sec:Schrodinger}. 

By \eqref{Zak equivalence}, the corresponding Wannier functions $w_a := \UZ^{-1} \phi_a$ satisfy $ \langle x \rangle^r w_a \in L^2(\R^d)$ for all $r \in \N$.  Then the set of all the translated functions  $\set{w_{\gamma, a}}$,  with $w_{\gamma, a}(x) = w_{a}(x - \gamma)$, provides a composite Wannier basis consisting of almost-exponentially localized functions, as stated in \ref{item:exp-localized}.

Moreover, $\Phi$ satisfies  \ref{tr} which in this context reads
$\phi_a(-k) = C \phi_a(k) = \overline{\phi_a}(k)$, since $C$ is just complex conjugation in $L^2(\T^d_Y)$. 
By Definition \ref{Def:CompositeWannierBasis}, one concludes that
\begin{equation*}
\overline{w_a}(x)  =   \frac{1}{|\B|^{1/2}} \int_{\B} \di k \,\eu^{ - \iu k \cdot x} \,\, \overline{\phi_a}(k, x)  
=  \frac{1}{|\B|^{1/2}} \int_{\B} \di k \,\eu^{  \iu (-k) \cdot x} \,\, \phi_a(- k, x)   = w_a(x), 
\end{equation*}
which yields property \ref{item:R-valued} and concludes the proof. 
\end{proof}


\bigskip

We sketch the structure of the proof of Theorem \ref{Thm:SmoothBlochFrames},  provided in Sections
\ref{Sec:Algorithm} and \ref{Sec:SmoothingProcedure}. 
First, one easily notices that, in view of properties $(\mathrm{F}_{2})$ and $(\mathrm{F}_{3})$, a global symmetric Bloch frame $\Phi: \R^d \to \Hi^m$
is completely specified by the values it assumes on the \emph{effective unit cell}
\[ \Bred := \set{k = \sum_j k_j e_j  \in \B: k_1 \ge 0}.  \]
Indeed, every point $k \in \R^d$ can be written (with an a.e.-unique decomposition) as $k = (-1)^s k' + \la$,  for some $k' \in \Bred$, $\la \in \Lambda$  and $s \in \set{0,1}$. Then the symmetric Bloch frame $\Phi$ satisfies 
$\Phi(k) = \tau_\la \Theta^s \, \Phi(k')$ for $k' \in \Bred$.
Viceversa,  a local Bloch frame $\cell{\Phi} : \Bred \to \Hi^m$ can be canonically extended to a global {symmetric} Bloch frame $\Phi$ by posing
\begin{equation} \label{Eq:Symmetric extension}
\Phi(k) = \tau_\la \, \Theta^s \, \cell{\Phi}(k')  \qquad \qquad \mbox{ for } k = (-1)^s k' + \la. 
\end{equation} 
However, to obtain a global \emph{continuous} Bloch frame,  the map $\cell{\Phi} : \Bred \to \Hi^m$ must satisfy some non-trivial \virg{gluing conditions} on the boundary $\partial \Bred$, involving vertices, edges (for $d \geq 2$), faces (for $d \geq 3$), and so on.  In Section \ref{Sec:Algorithm}, we  investigate in detail such conditions, showing that it is always possible to construct a local continuous Bloch frame $\cell{\Phi} :  \Bred \rightarrow \Hi^m$ satisfying them, provided $d \leq 3$. 
More specifically, we assume as given a continuous Bloch frame $\Psi :  \Bred \rightarrow \Hi^m$, called \emph{the input frame},
which does not satisfy any special condition on the boundary of $\Bred$, as \eg the outcome of numerical computations in solid-state physics. Then we explicitly construct a unitary matrix $\cell{U}(k)$ such that the \virg{corrected} frame
$$
\cell{\Phi}(k)_a  := \sum_{b =1}^{m} \Psi_b(k) \, \cell{U}(k)_{ba}
$$ 
is still continuous and satisfies all the relevant symmetry conditions on $\partial \Bred$.  Then, formula \eqref{Eq:Symmetric extension}  will provide a  global \emph{continuous symmetric} Bloch frame $\Phi$. \newline
A na\"{i}f smoothing procedure,  based on the Steenrod's Approximation Theorem, starting from $\Phi$ would yield a {global smooth $\tau$-equivariant} Bloch frame which, in general, does not satisfies property $(\mathrm{F}_{3})$.  For this reason, we develop in Section  \ref{Sec:SmoothingProcedure}  a new symmetry-preserving smoothing algorithm which, starting from a {global continuous symmetric} Bloch frame,  produces a global \emph{smooth symmetric} Bloch frame arbitrarily close to the former one (Theorem \ref{Lem:SmoothingProcedure}).  The latter procedure, which holds true in any dimension, yields the global smooth symmetric Bloch frame  whose existence is claimed in Theorem \ref{Thm:SmoothBlochFrames}. 


\begin{rmk}[Systems with fermionic time-reversal symmetry] \label{Rem:FermionicSystems}
Our results heavily rely on the fact that we consider a \emph{bosonic} (or even) time-reversal (TR) symmetry. 
In other instances, as in the context of TR-symmetric topological insulators \cite{HasanKane},  and specifically 
in the Kane-Mele model \cite{KaneMele2005}, assumption \ref{item:TRS} is replaced by 
\begin{enumerate}[label=$(\mathrm{P}_{\arabic*, -})$,ref=$(\mathrm{P}_{\arabic*,-})$]
\setcounter{enumi}{2}
\item \label{item:fermionicTRS} \emph{fermionic time-reversal symmetry}: there exists an antiunitary operator $\Theta$ acting on $\Hi$ such that 
\[ P(-k) = \Theta \, P(k) \, \Theta^{-1} \quad \text{and} \quad \Theta^2 = -  \Id_{\Hi}. \]
\end{enumerate}
Then the statement analogous to Theorem \ref{Thm:SmoothBlochFrames} is false: there might be topological  obstruction to the existence of a continuous symmetric Bloch frame \cite{FuKane, GrafPorta}. One proves  \cite{GrafPorta, FiMoPa} that this obstruction is classified by a $\Z_2$ topological invariant for $d=2$, and by four $\Z_2$ invariants for $d =3$, and that the latter equal the indices introduced by Fu, Kane and Mele \cite{FuKane, FuKaneMele}. However, if one does not require time-reversal symmetry but only $\tau$-equivariance,  then a global smooth Bloch frame does exist  even in the fermionic case,  as a consequence of the vanishing of the first Chern class and of the result in \cite{Panati2007}, see \cite{MoPa} for a detailed review.    
\hfill  $\lozenge$ 
\end{rmk}

\newpage
\section{Proof:  Construction of a smooth symmetric Bloch frame} \label{Sec:Algorithm}

In this Section, we provide an explicit algorithm to construct a global smooth symmetric Bloch frame,  
as claimed in  Theorem \ref{Thm:SmoothBlochFrames}. 

Our general strategy  will be the following. 
We consider a local continuous  (resp.\ smooth)
\footnote{A smooth input frame $\Psi$ is required only to write an explicit formula for the continuous extension from 
the boundary $\partial \Bred$ to the whole $\Bred$,  as detailed in Remarks \ref{Rem:ExplicitFilling 2d} and  \ref{Rem:ExplicitFilling 3d}. At a first reading, the reader might prefer to focus on the case of a continuous input frame. 
}\  
Bloch frame $\Psi : \Bred \to \Hi^m$, which  always exists since $\Bred$ is contractible and no special conditions on the boundary are imposed.
\footnote{Moreover, a smooth $\Psi$ can be explicitly constructed by using the intertwining unitary by Kato and Nagy \cite[Sec. I.6.8]{Kato}  on finitely-many sufficiently small open sets covering  $\Bred$.  In the applications to computational physics, 
$\Psi$ corresponds to the outcome of the numerical diagonalisation of the Hamiltonian at fixed crystal momentum, 
followed by a choice of quasi-Bloch functions and by a standard routine which corrects the phases to obtain a (numerically) 
continuous (resp.\ smooth) Bloch frame on $\Bred$.
}\  
We look for a unitary-matrix-valued  map
${U}: \Bred \to \U(\C^m)$ such that the modified local Bloch frame 
\begin{equation} \label{Eq: healer_acts}
{\Phi}_a(k) =  \sum_{b=1}^m  \Psi_b(k) {U}_{ba}(k), \quad {U}(k) \in \U(\C^m)
\end{equation}
satisfies \ref{item:F2} and \ref{item:F3} on the boundary $\partial \Bred$. The latter requirement 
corresponds to conditions on the values that ${U}$ assumes  
on the vertices, edges and faces of $\partial \Bred$, according to the dimension. These 
conditions  will be investigated in the next Subsections, after a preliminary characterization of the relevant symmetries.

\subsection{The relevant group action}
\label{Sec:Group_action}

Properties  \ref{item:tau} and \ref{item:TRS} are related to some fundamental automorphisms of $\R^d$, namely the maps  $c$ and $t_{\lambda}$ defined by 
\begin{equation} \label{c t_la}
c(k) = - k  \qquad \mbox{ and } \qquad  t_{\lambda}(k) = k + \lambda  \quad \mbox{ for } \la \in \Lambda. 
\end{equation} 
Since $c \, t_{\lambda} = t_{- \lambda} c$ and $c^2 = t_0$, one concludes that the relevant symmetries are encoded in the group \begin{equation} \label{G_d}
G_d := \set{t_\lambda, t_\lambda c}_{\lambda \in \Lambda}  \subset \mathrm{Aut}(\R^d). 
\end{equation}

\bigskip

We notice that, assuming also \ref{item:TRtau}, the action of $G_d$ on $\R^d$ can be lifted to an action on 
$\R^d \times \Fr(m,\Hi)$,  where $\Fr(m,\Hi)$ is the set of orthonormal $m$-frames in $\Hi$. 
To streamline the notation, we denote by $\Phi = (\phi_1, \ldots, \phi_m)$ an element of $\Fr(m, \Hi)$. 
Any bounded linear or antilinear operator $A : \Hi \to \Hi$ acts on frames componentwise, \ie we set 
$$   
A \Phi :=  (A\phi_1, \ldots, A\phi_m).
$$
Moreover, the space $\Fr(m, \Hi)$ carries a free right action%
\footnote{This terminology means that $\Phi \act \id = \Phi$, $(\Phi \act U_1) \act U_2 = \Phi \act (U_1 U_2)$ and that if $\Phi \act U_1 = \Phi \act U_2$ then $U_1 = U_2$, for all $\Phi \in \Fr(m,\Hi)$ and $U_1, U_2 \in \U(\C^m)$.} %
of the group  $\U(\C^m)$, denoted by
\[ (\Phi \act U)_b := \sum_{a = 1}^{m} \phi_a \, U_{ab}. \]
A similar notation appears  in \cite{GrafPorta}.  Notice that, by the antilinearity of the time-reversal operator $\Theta$, 
one has
\[ \Theta (\Phi \act U) = (\Theta \, \Phi) \act \overline{U}, \quad \text{for all } \Phi \in \Fr(m,\Hi), \: U \in \U(\C^m). \]

\noindent A lift of the $G_d$ action from $\R^d$ to $\R^d \times \Fr(m,\Hi)$ is obtained by considering, as generators, the automorphisms $C$ and $T_{\la}$, defined by  
\begin{equation} \label{G_lifted}
C (k, \Phi) =  (c(k), \Theta \, \Phi)  \qquad  T_\la (k, \Phi) = (t_{\la}(k), \tau_{\la} \Phi)
\end{equation}
for any $(k, \Phi) \in \R^d \times \Fr(m,\Hi)$.  The relation  $t_{\lambda} \, c = c \, t_{- \lambda}$ implies that, for every $\la \in \Lambda$, one has to impose the relation
\begin{eqnarray*}
T_{\la} C (k, \Phi) &=& C T_{-\la}(k, \Phi)              \qquad  \forall (k, \Phi) \in \R^d \times \Fr(m,\Hi)     \\
\mbox{\ie } \tau_\la \Theta \, \Phi  &=& \Theta \tau_{-\la} \, \Phi    \qquad \qquad  \forall \Phi \in \Fr(m,\Hi), 
\end{eqnarray*}
which holds true in view of $(\mathrm{P}_{4})$. Thus the action of $G_d$ is lifted to $\R^d \times \Fr(m, \Hi)$. 

\bigskip

Given a family of projectors $\mathcal{P}$  satisfying  \ref{item:tau}, \ref{item:TRS}  and \ref{item:TRtau}, it is natural to consider the  set of global Bloch frames for $\mathcal{P}$, here denoted by $\Fr(\mathcal{P})$. Notice that
$$
\Fr(\mathcal{P}) \subset  \set{f : \R^d \rightarrow \Fr(m, \Hi) } \subset \R^d \times \Fr(m, \Hi).
$$ It is easy to check that the 
action of $G_d$, previously extended to $\R^d \times \Fr(m,\Hi)$,  restricts to $\Fr(\mathcal{P})$. Indeed, whenever $ \Phi$ is an orthonormal frame for $\Ran P(k)$ one has that 
\begin{eqnarray*}
P(t_\la(k)) \tau_\la \Phi & =& \tau_\la P(k) \tau^{-1}_\la \, \tau_\la \Phi  = \tau_\la P(k) \Phi = \tau_\la  \Phi,\\
P(c(k)) \Theta \, \Phi & = &   \Theta P(k) \Theta^{-1} \,  \Theta \, \Phi  =     \Theta P(k) \Phi =  \Theta \, \Phi, 
\end{eqnarray*}
yielding that  $\tau_\la  \Phi$ is an orthonormal frame in  $\Ran(P(t_\la(k)))$ and 
$ \Theta \, \Phi$  is an orthonormal frame in $\Ran(P(c(k)))$.

\subsection{Solving the vertex conditions} 
\label{Sec:VertexConditions}

The relevant vertex conditions are associated to those points $k \in \R^d$ which have a non-trivial stabilizer with respect to the 
action of $G_d$, namely to the points in the set
\begin{equation} \label{Def:vertices}
V_d = \set{k \in \R^d :  \exists \, g \in G_d,  g \neq \Id :  g(k) = k}.
\end{equation}
Since $G_d= \set{t_\la, t_\la \, c}_{\la \in \Lambda}$ and $t_\la$ acts freely on $\R^d$, the previous definition reads
\begin{eqnarray*} 
V_d  &=& \set{k \in \R^d :  \exists \la \in \Lambda  : \,\, t_{\la} c (k) = k} \\
         &=& \set{k \in \R^d :  \exists \la \in \Lambda : - k + \la = k} = \set{\half \la}_{\la \in \Lambda},
\end{eqnarray*}
\ie $V_d$ consists of those points
\footnote{In the context of topological insulators, such points are called \emph{time-reversal invariant momenta} (TRIMs)
in the physics literature.}
which have half-integer coordinates with respect to the basis $\set{e_1, \ldots, e_d}$. 
For convenience, we set $k_\la := \half \la$. 

\goodbreak

If $\Phi$ is a symmetric Bloch frame, then conditions $(\mathrm{F}_{2})$ and $(\mathrm{F}_{3})$ imply that
\begin{equation} \label{VertexCondition} 
\Phi(k_\lambda) = \Phi(t_\lambda c(k_\lambda)) =  \tau_\la \Theta \, \Phi(k_\lambda) \qquad \qquad k_{\la} \in V_d.
\end{equation}
We refer to \eqref{VertexCondition} as the \textbf{vertex condition at the point $k_\la \in V_d$}. For a generic Bloch frame $\Psi$, instead, $\Psi(k_\lambda)$ and $\tau_\lambda \Theta \Psi(k_\lambda)$ are different. Since they both are orthonormal frames in $\Ran P(k_\lambda)$,  there exists a unique unitary matrix $U\sub{obs}(k_\lambda) \in \U(\C^m)$ such that
\begin{equation} \label{UobsVertices}
\Psi(k_\lambda) \act U\sub{obs}(k_\lambda) = \tau_\lambda \Theta \Psi(k_\lambda),  \qquad \la \in \Lambda.
\end{equation}
The obstruction unitary $U\sub{obs}(k_\lambda)$ must satisfy a compatibility condition. Indeed, by applying $\tau_\lambda \Theta$ to both sides of \eqref{UobsVertices} one obtains
\begin{align*}
\tau_\lambda \Theta \left( \Psi(k_\lambda) \act U\sub{obs}(k_\lambda) \right) & = \tau_\lambda \Theta \tau_\lambda \Theta \Psi(k_\lambda)  \\
& = \tau_\lambda \tau_{-\lambda} \Theta^2 \Psi(k_\lambda)  = \Psi(k_\lambda)
\end{align*}
where assumption \ref{item:TRtau} has been used. On the other hand, the left-hand side also reads
\begin{align*}
\tau_\lambda \Theta \left( \Psi(k_\lambda) \act U\sub{obs}(k_\lambda) \right) & = \( \tau_\lambda \Theta \Psi(k_\lambda) \) \act \overline{U\sub{obs}}(k_\lambda) = \\
& = \Psi(k_\lambda) \act \left( U\sub{obs}(k_\lambda) \, \overline{U\sub{obs}}(k_\lambda) \right).
\end{align*}
By the freeness of the action of $\U(\C^m)$ on frames, one concludes that  
\begin{equation} \label{Compatibility}
U\sub{obs}(k_\lambda) =  U\sub{obs}(k_\lambda)\Tr 
\end{equation}
where $M\Tr$ denotes the transpose of the matrix $M$.

The value of the  unknown $U$, appearing in \eqref{Eq: healer_acts}, at the point $k_\la \in V_d$ is constrained by the value of the obstruction matrix $U\sub{obs}(k_\lambda)$. 
Indeed, from  \eqref{VertexCondition} and \eqref{Eq: healer_acts} it follows that for every $\la \in \Lambda$
\begin{align*}
\Psi(k_\lambda) \act U(k_\lambda) & = \Phi(k_\lambda) = \tau_\lambda \Theta \, \Phi(k_\lambda)  \\
& = \tau_\lambda \Theta \left( \Psi(k_\lambda) \act U(k_\lambda) \right) \\
& =  \( \tau_\lambda \Theta \Psi(k_\lambda) \) \act  \overline{U}(k_\lambda) \\
& = \Psi(k_\lambda) \act U\sub{obs}(k_\lambda) \overline{U}(k_\lambda).
\end{align*}
By the freeness of the $\U(\C^m)$-action, we obtain the  condition
\footnote{The presence of the transpose in condition \eqref{Uobs-U} might appear unnatural in the context of our Assumptions. 
A more natural reformulation of  condition \eqref{Uobs-U}, involving an orthogonal structure canonically associated to $\Theta$, will be discussed in a forthcoming paper \cite{CFMP}.}
\begin{equation} \label{Uobs-U}
U\sub{obs}(k_\lambda) = U(k_\lambda)  U(k_\lambda)\Tr.
\end{equation}

\noindent The existence of a solution $U(k_\lambda) \in \U(\C^m)$ to equation \eqref{Uobs-U} is granted by the following Lemma, 
which can be applied to $V = U\sub{obs}(k_\lambda)$ in view of \eqref{Compatibility}. 

\begin{lemma}[Solution to the vertex equation] \label{V->U}
Let $V \in \U(\C^m)$ be such that $V \Tr = V$.  Then there exists a unitary matrix $U \in \U(\C^m)$ such that $V = U \, U\Tr $.
\end{lemma}
\begin{proof} Since $V \in  \U(\C^m)$ is normal, it can be unitarily diagonalised. Hence, there exists a unitary matrix $W \in \U(\C^m)$ 
such that
$$
V = W \eu^{\iu M} W^* 
$$
where $M = \diag(\mu_1, \ldots, \mu_m)$ and each $\mu_j$ is chosen so that
\footnote{The latter condition is crucial: it expresses the fact that the arguments of the eigenvalues $\set{\omega_1, \ldots, \omega_m}$ of  $V$ are \virg{synchronized}, \ie  they are computed by using the same branch of the complex logarithm.}
$\mu_j \in [0, 2\pi)$.  We set 
$$
U = W \eu^{\iu M/2} W^*. 
$$
Since $V\Tr = V$  one has  $\overline{W} \eu^{\iu M} W\Tr  = W \eu^{\iu M} W^* $, yielding 
$$
\eu^{\iu M} W\Tr W = W\Tr W \eu^{\iu M},
$$
\ie the matrix $\eu^{\iu M}$ commutes with $A := W\Tr W$.  Thus also $\eu^{\iu M/2}$ commutes with $A$ (since each $\mu_j$ is in $[0, 2 \pi)$), hence one has
\[ U  U\Tr  = 
W \eu^{\iu M/2} A^{-1} \eu^{\iu M/2} A W^* = W \eu^{\iu M} W^* = V. \]
\end{proof}
In view of \eqref{Compatibility}, we have the following
\begin{crl} \label{Cor:U_factorization}
For every $k_\la \in V_d$ there exists a unitary matrix $U(k_\lambda)$ such that
$ U\sub{obs}(k_\lambda) = U(k_\lambda) U(k_\lambda)\Tr$.
In particular, the Bloch frame $\Phi(k_{\la}) = \Psi(k_{\la}) \act U(k_{\la})$ satisfies the vertex 
condition \eqref{VertexCondition} at the point $k_{\la} \in V_d$.
\end{crl}

\goodbreak

\subsection{Construction in the $1$-dimensional case}
\label{Sec:1_dimensional}

In the $1$-dimensional case, the boundary of $\Bred$ consists of two vertices $v_0 = 0$ and $v_1 = k_{e_1}$. 
Given, as an input, a {continuous} Bloch frame $\Psi: \Bred \to \Hi^m$,  
equation \eqref{UobsVertices} provides,  for each vertex, an obstruction matrix $U\sub{obs}(v_i)$. In view of Corollary 
\ref{Cor:U_factorization}, one obtains a unitary $U(v_i)$ which solves equation \eqref{Uobs-U} for $k_\la =v_i$, $i \in \set{0,1}$. 

Since $\U(\C^m)$ is a path-connected manifold, there exists a smooth path $W:[0, \half] \to \U(\C^m)$ 
such that  $W(0) = U(v_0)$ and $W(\half) = U(v_1)$. Moreover, the path $W$ can be explicitly constructed, as detailed in the following Remark.  

\begin{rmk}[Interpolation of unitaries] \label{Rem:Interpolation}
The problem of constructing a  smooth interpolation between two unitaries $U_1$ and $U_2$ in $\U(\C^m)$ has an easy explicit solution.  First, by left multiplication times $U_1^{-1}$, the problem is equivalent to the construction of a smooth interpolation between $\id$ and $U_1^{-1}U_2 =: U_{*} \in \U(\C^m)$. Since $U_{*}$ is normal, there exists a unitary matrix $S_*$ such that $S_* U_* S_*^{-1} = \eu^{\iu D}$, with $D = \diag(\delta_1, \ldots, \delta_m)$ a diagonal matrix. Then the map
$t \mapsto W(t) := S_*^{-1} \eu^{\iu 2t D} S_*$ is an explicit smooth interpolation between $W(0)= \id$ and $W(\half)= U_*$. 
\hfill $\lozenge$
\end{rmk}

We define a local continuous Bloch frame $\cell{\Phi}:  \Bred \to \Hi^m$ by setting 
$$
\cell{\Phi}(k) = \Psi(k) \act W(k),  \qquad \qquad  k \in \Bred.
$$
Notice that $\cell{\Phi}$ satisfies, in view of the construction above, the vertex conditions
\begin{equation} \label{Eq:1d_vertices}
\cell{\Phi}(0) = \Theta \, \cell{\Phi}(0)    \quad  \mbox{ and } \quad  \cell{\Phi}(k_{e_1}) = \tau_{e_1} \Theta \, \cell{\Phi}(k_{e_1}),
\end{equation}
which are  special cases of condition \eqref{VertexCondition}. 
We extend $\cell{\Phi}$ to a global Bloch frame $\Phi: \R^1 \to \Hi^m$ by using equation \eqref{Eq:Symmetric extension}. 
We claim that $\Phi$ is a \emph{continuous} symmetric Bloch frame. Indeed, it satisfies \ref{item:F2} and 
\ref{item:F3}  in view  of \eqref{Eq:Symmetric extension} and it is continuous since $\cell{\Phi}$ satisfies \eqref{Eq:1d_vertices}. 
On the other hand, $\Phi$  is in general non-smooth at the vertices in $V_1$. By using the symmetry-preserving smoothing procedure, as stated in Proposition \ref{Lem:SmoothingProcedure}, we obtain a global smooth symmetric Bloch frame $\Phi\sub{sm}$. This concludes the proof of Theorem \ref{Thm:SmoothBlochFrames} for $d=1$.  


\subsection{Construction in the $2$-dimensional case}
\label{Sec:2_dimensional}

The reduced unit cell $\Bred$ contains exactly six elements in $V_2$.  In adapted coordinates, so that 
$(k_1, k_2)$ represents the point $k_1 e_1 + k_2 e_2$ for $\Lambda = \mathrm{Span}_{\Z}\set{e_1, e_2}$, they are  labelled as follows (Figure \ref{fig:Bred}): 
\begin{equation} \label{vertices}
\begin{aligned}
v_1 = (0,0), \quad v_2 = \left( 0, - \frac{1}{2} \right), \quad v_3 = \left( \frac{1}{2}, - \frac{1}{2} \right), \\
v_4 = \left(\frac{1}{2} ,0 \right), \quad v_5 = \left( \frac{1}{2}, \frac{1}{2} \right), \quad v_6 = \left( 0, \frac{1}{2} \right).
\end{aligned}
\end{equation}
The oriented segment joining $v_i$ to $v_{i+1}$ (with $v_7 \equiv v_1$) is labelled by $E_i$. 


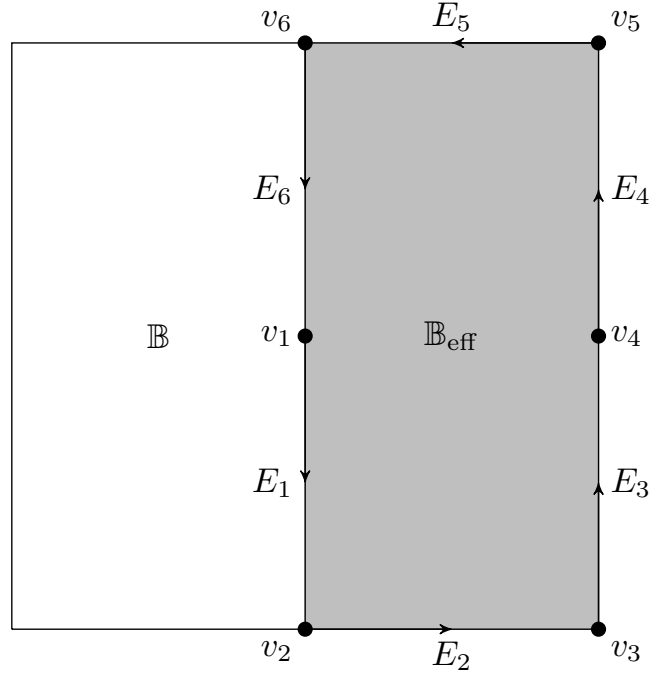
\begin{figure}[ht]
\centering
\scalebox{1.3}{%
\begin{footnotesize}
\begin{tikzpicture}[>=stealth']
\filldraw [lightgray] (0,-3) -- (3,-3) -- (3,3) -- (0,3) -- (0,-3);
\draw (-3,-3) -- (-3,3) -- (3,3) -- (3,-3) -- (-3,-3);
\draw (0,-3) -- (0,3);
\draw [->] (0,0) -- (0,-1.5) node [anchor=east] {$E_1$};
\draw [->] (0,-3) -- (1.5,-3) node [anchor=north] {$E_2$};
\draw [->] (3,-3) -- (3,-1.5) node [anchor=west] {$E_3$};
\draw [->] (3,0) -- (3,1.5) node [anchor=west] {$E_4$};
\draw [->] (3,3) -- (1.5,3) node [anchor=south] {$E_5$};
\draw [->] (0,3) -- (0,1.5) node [anchor=east] {$E_6$};
\filldraw [black] (0,0) circle (2pt) node [anchor=east] {$v_1$}
					(0,-3) circle (2pt) node [anchor=north east] {$v_2$}
					(3,-3) circle (2pt) node [anchor=north west] {$v_3$}
					(3,0) circle (2pt) node [anchor=west] {$v_4$}
					(3,3) circle (2pt) node [anchor=south west] {$v_5$}
					(0,3) circle (2pt) node [anchor=south east] {$v_6$};
\draw (1.5,0) node {$\Bred$}
      (-1.5,0) node {$\B$};
\end{tikzpicture}
\end{footnotesize}
} 
\caption{The effective unit cell (shaded area), its vertices and its edges. We use adapted coordinates $(k_1,k_2)$ such that $k = k_1 e_1 + k_2 e_2$.}
\label{fig:Bred}
\end{figure}


We start from a local continuous (resp.\ smooth) Bloch frame $\Psi: \Bred \to \Hi^m$.  Given $\Psi$, the obstruction matrix defined in \eqref{UobsVertices} yields, via Corollary \ref{Cor:U_factorization},  a unitary matrix $U(v_i)$ solving equation \eqref{Uobs-U}, for $i \in \set{1, \ldots, 4}$. 

\goodbreak

\subsubsection{Construction of the frame on the $1$-skeleton}
\label{Sec:1-skeleton}

As in the $1$-dimensional case, we exploit the constructive existence of a smooth path $W_i : [0, 1/2] \to \U(\C^m)$ such that 
$W_i(0)= U(v_i)$ and $W_i(\half) = U(v_{i+1})$, for  $i \in \set{1, 2,  3}$. These $\U(\C^m)$-valued paths are concatenated by setting  
\[ \widetilde{U}(k) := \begin{cases}
W_1(-k_2) & \text{if } k \in E_1, \\
W_2(k_1) & \text{if } k \in E_2, \\
W_3(k_2+1/2) & \text{if } k \in E_3,
\end{cases} \]
so to obtain a piecewise-smooth map $\widetilde{U} : E_1 \cup E_2 \cup E_3 \to \U(\C^m)$. Let 
$$
\widetilde{\Phi}(k) := \Psi(k) \act \widetilde{U}(k)  \qquad \mbox{ for } k \in E_1 \cup E_2 \cup E_3. 
$$
We extend the map $\widetilde{\Phi}$ to a continuous (resp.\  piecewise-smooth) symmetric Bloch frame $\widehat{\Phi}$ on $\partial \Bred$ by imposing properties  \ref{item:F2} and  \ref{item:F3},   \ie by setting
\begin{equation} \label{Phihat}
\widehat{\Phi}(k) := \begin{cases}
\widetilde{\Phi}(k) & \text{if } k \in E_1 \cup E_2 \cup E_3 \\
\tau_{e_1} \Theta \, \widetilde{\Phi}(t_{e_1}c (k)) & \text{if } k \in E_4 \\
\tau_{e_2} \, \widetilde{\Phi}(t^{-1}_{e_2}(k)) & \text{if } k \in E_5 \\
\Theta \, \widetilde{\Phi}(c(k)) & \text{if } k \in E_6.
\end{cases}
\end{equation}

By construction $\widehat{\Phi}$ satisfies all the \emph{edge symmetries} for a symmetric Bloch frame $\Phi$ listed below:
\begin{equation} \label{EdgeSymmetries}
\begin{aligned}
&\Phi(c(k)) = \Theta \, \Phi(k)  && \text{for } k \in E_1 \cup E_6 \\
&\Phi(t_{e_2}(k)) = \tau_{e_2} \Phi(k) && \text{for } k \in E_2\\
&\Phi(t_{e_1} c(k)) = \tau_{e_1} \Theta \, \Phi(k) \quad && \text{for } k \in E_3 \cup E_4\\
&\Phi(t^{-1}_{e_2}(k)) = \tau^{-1}_{e_2} \Phi(k) && \text{for } k \in E_5.
\end{aligned}
\end{equation}

The map $\widehat{\Phi}: \partial \Bred \to \Hi^m$ is {continuous} (resp.\ piecewise-smooth), since $\widetilde{\Phi}$ is continuous (resp.\ piecewise-smooth) and satisfies by construction the vertex conditions at $v_i$  for $i \in \set{1, 4}$.


\subsubsection{Extension to the $2$-torus }
\label{Sec:2-skeleton}

Since both $\widehat{\Phi}(k)$ and the input frame $\Psi(k)$ are orthonormal frames in $\Ran P(k)$, for every $k \in \partial \Bred$, 
there exists  a unique unitary matrix $\widehat{U}(k)$ such that
\begin{equation} \label{hatU}
\widehat{\Phi}(k) = \Psi(k) \act \widehat{U}(k) \qquad \text{ for } k \in \partial\Bred.
\end{equation}
Explicitly, $ \widehat{U}(k)_{ab} = \inner{\psi_a(k)}{\widehat \phi_b(k)}$, which also show that the map $\widehat{U}$ is {continuous} (resp.\ piecewise-smooth) on $\partial \Bred$. 

We look for a {continuous} extension $\cell{U} : \Bred \to \U(\C^m)$ of $\widehat{U}$, such that 
$\cell{\Phi} := \Psi \act \cell{U} $ satisfies the edge symmetries \eqref{EdgeSymmetries}. 
Noticing that $\partial \Bred$ is homeomorphic to a circle $S^1$, we use some well-known facts in algebraic topology: 
%
if $X$ is a topological space, then a continuous map $f  : S^1 \to X$   extends to a continuous map $F  :   D^2 \to X$, where $D^2$ is the $2$-dimensional disc enclosed by the circle $S^1$, if and only if its homotopy class $[f]$ is the trivial element in $\pi_1(X)$.  Since, in our case, the space $X$ is the group $\U(\C^m)$, we also use the fact that the exact sequence of groups
\[ 1 \longrightarrow \mathcal{S}\U(\C^m) \longrightarrow \U(\C^m) \xrightarrow{\det} U(1) \longrightarrow 1 \]
induces an isomorphism $\pi_1(\U(\C^m)) \simeq \pi_1(U(1))$. On the other hand, the degree homomorphism
\begin{equation} \label{deg} 
\deg : \pi_1(U(1)) \stackrel{\sim}{\longrightarrow} \Z, \quad [\varphi : S^1 \rightarrow U(1)] \longmapsto \frac{1}{2\pi \iu} \oint_{S^1} \,  \varphi(z)^{-1} \partial_z \varphi(z) \, \di z
\end{equation}
establishes an isomorphism of groups $\pi_1(U(1)) \simeq \Z$. 
We conclude that a continuous map $f  :   \partial \Bred \to \U(\C^m)$ can be continuously extended to $F  :   \Bred \to \U(\C^m)$ if and only if $\deg([\det f]) \in \Z$ is zero.


The following Lemma is the crucial step in the $2$-dimensional construction. It shows that, even if $\deg([\det \widehat U]) = r \neq 0$, it is always possible to construct a continuous map $X : \partial \Bred \to \U(\C^m)$ such that  $\deg ([\det \widehat U X]) = 0$ and $\widehat \Phi \act X$ still satisfies the edge symmetries.\footnote{This is a special feature of systems with \emph{bosonic} TR-symmetry: if assumption \ref{item:TRS} is replaced by \ref{item:fermionicTRS}, the analogous statement does not hold true \cite{FiMoPa}.} 

\begin{lemma}[Solution to the face-extension problem] \label{Lem:Xmatrix}
Let $r \in \Z$.  There exists a  piecewise-smooth map $X : \partial \Bred \to \U(\C^m)$ such that:
\begin{enumerate}[label=$(\mathrm{\roman*})$]
\item $\deg([\det X]) = - r$;
\item \label{item:ii} if a Bloch frame $\Phi$ satisfies the edge symmetries  \eqref{EdgeSymmetries}, the frame $\Phi \act X$ also does;
\item \label{item:iii} $X(k) \neq \id$ only for $k \in E_3 \cup E_4$.
\end{enumerate}
\end{lemma}

Property \ref{item:iii} will not be used in this Section, but it will be useful to solve the $3$-dimensional problem. 

\begin{proof}  First, we translate  (ii) into an explicit condition on $X$. 
For $k \in E_1 \cup E_6$, condition (ii) means that for every $\Phi$ such that $\Phi(-k) = \Theta \, \Phi(k)$ one has that
\begin{align*}
\(\Phi \act X \) (-k) & = \Theta \(\Phi \act X \)(k)  \\
& \Updownarrow \\
\Phi(-k) \act X(-k) & = \Theta \left( \Phi(k) \act X(k) \right) \\
& \Updownarrow \\
\( \Theta \, \Phi(k) \) \act X(-k) & = \( \Theta \, \Phi(k) \) \act \overline{X}(k),
\end{align*}
yielding the explicit condition 
\begin{equation} \label{E1E6}
X(-k) = \overline{X}(k),   \quad k \in E_1 \cup E_6.
\end{equation}
Similarly,  one obtains
\begin{eqnarray}     
\label{E3E4}     X(t_{e_{1}} c(k)) &=& \overline{X}(k)  \qquad \qquad  \mbox{for } k \in E_3 \cup E_4 \\
\label{E2}        X(t_{e_{2}}(k)) &=& X(k)  \qquad \qquad \mbox{for } k \in E_2 \\
\label{E5}        X(t^{{-1}}_{e_{2}}(k)) &=& X(k)  \qquad \qquad \mbox{for } k \in E_5.  
\end{eqnarray}
Thus  condition (ii) on $X$ is equivalent to the relations \eqref{E1E6}, \eqref{E3E4}, \eqref{E2} and \eqref{E5}. 


We now exhibit a map $X : \partial \Bred \to \U(\C^m)$ which satisfies the previous relations, and such that  $\deg([\det X]) = - r$.  Define  $\xi : \partial \Bred \to \C$ by 
\begin{equation} \label{Def:x_{map}}
\xi(k) :=  \begin{cases} 
\,\, \eu^{{- \iu 2 \pi r (k_{2} + \half) }}    & \qquad \text{for } k \in E_3 \cup E_4 \\
\,\, 1                                       & \qquad \text{otherwise},
\end{cases}
\end{equation}
and set $X(k) := \diag(\xi(k), 1, \ldots, 1) \in \U(\C^{m})$ for $k \in \partial \Bred$. The map $X$ is clearly 
 piecewise-smooth.  
Then, one easily checks that: 
\begin{enumerate}[label=$(\mathrm{\roman*})$]
\item $\deg([\det X]) = - r$, since
$
\deg([\det X]) =  \deg([\xi]) = -r.  
$
\item $X$ trivially satisfies relations \eqref{E1E6}, \eqref{E2} and \eqref{E5}, since $X(k) \equiv \id$ for $k \in E_1 \cup E_2 \cup E_5 \cup E_6$.  It also satisfies  relation \eqref{E3E4}. Indeed, let $k = (\half, k_2) \in E_3 \cup E_4$.  Since 
$t_{e_1} c (k) = (\half, - k_2)$, one has 
\begin{eqnarray*} 
X(t_{e_1}c(k)) = X(\half, - k_2) &=&  \diag(\xi(\half, - k_2), 1, \ldots, 1) \\
&=&  \diag(\overline{\xi(\half, k_2)}, 1, \ldots, 1) = \overline{X}(k).  
\end{eqnarray*}
\item property (iii) is satisfied by construction. 
\end{enumerate}
\end{proof}

Set $r := \deg([\det \widehat{U}])$. In view of Lemma \ref{Lem:Xmatrix},  the {continuous}  (resp.\ piecewise-smooth) map $U := \widehat{U} X : \partial \Bred \to \U(\C^m)$ satisfies $\deg([\det U]) = 0$ and hence extends to a {continuous} (resp.\ piecewise-smooth) map $\cell{U} : \Bred \to \U(\C^m)$. 
Moreover, the extension procedure is explicitely constructive whenever $U$ is piecewise-smooth,  as detailed in Remark \ref{Rem:ExplicitFilling 2d}.
By setting  $\cell{\Phi}(k) := \Psi(k) \act \cell{U}(k)$, we obtain a  {continuous} symmetric Bloch frame on the whole reduced unit cell $\Bred$,  which moreover satisfies the edge symmetries \eqref{EdgeSymmetries} in view of item \ref{item:ii} in Lemma \ref{Lem:Xmatrix}. 
Then formula \eqref{Eq:Symmetric extension} defines a global symmetric Bloch frame $\Phi$, which is {continuous} in view of the fact that $\cell{\Phi}$ satisfies \eqref{EdgeSymmetries}. The symmetry-preserving smoothing procedure (Proposition \ref{Lem:SmoothingProcedure}) yields a global {smooth symmetric} Bloch frame, arbitrarily close to $\Phi$. 
This concludes the proof of Theorem \ref{Thm:SmoothBlochFrames} for $d=2$.  


\begin{rmk}[Explicit extension to the whole effective cell, $d=2$] \label{Rem:ExplicitFilling 2d}
We emphasize that the extension of the \emph{piecewise-smooth} map $U :  \partial \Bred \to \U(\C^m)$, with $\deg [\det U] = 0$, to a continuous (actually, piecewise-smooth) map $\cell{U} : \Bred \to \U(\C^m)$ is explicit. For notational convenience, 
we use the shorthand  $\U(m) \equiv \U(\C^m)$. 

First notice that the problem of constructing a continuous extension of $U$ can be decomposed into two simpler problems, 
since $\U(m) \approx \U(1) \times \mathcal{SU}(m)$ (as topological spaces), where the identification is provided \eg 
by the map
$$
W \mapsto (\det W, W^{\flat}) \in \U(1) \times \mathcal{SU}(m) \qquad \text{ with } W^{\flat} = \diag(\det W^{-1}, 1, \ldots,1) W. 
$$
Thus the problem is reduced to exhibit a continuous extension of 
(a) the map $f : k \mapsto \det U(k) \in \U(1)$, and (b) the map  $f^{\flat} : k \mapsto {U}^{\flat}(k) \in \mathcal{SU}(m)$.

\medskip

As for problem (a),  let $f:\partial\Bred\to \U(1)$ be a degree-zero piecewise-smooth function. 
Then, a  piecewise-smooth extension $F:\Bred\to \U(1)$ is constructed as follows. 
Let $\theta_0\in \mathbb{R}$ be such that $f(0,-1/2)=\eu^{\iu 2\pi  \theta_0}$. Define the piecewise-smooth function $\varphi:[0,3]\to \U(1)$ as

\begin{equation} \label{Eq:Phase}
\varphi(t)=
\begin{cases}
f(t, - 1/2)& \text{if $0\leq t\leq 1/2$}\\
f(1/2, - 1 + t)& \text{if $1/2\leq t\leq 3/2$}\\
f(2 - t, 1/2)& \text{if $3/2\leq t\leq 2$}\\
f(0, 5/2 - t)& \text{if $2\leq t\leq3$}\\
\end{cases}
\end{equation}
and set
\[
\theta(t)=\theta_0+\frac{1}{2\pi \iu}\int_0^t\varphi(\tau)^{-1}\varphi'(\tau)\di \tau,    \qquad \text{ for } t \in[0,3].
\]

\medskip

\noindent By the Cauchy integral formula, $\theta(3)=\theta(0)+\deg(f)=\theta(0)=\theta_0$. Moreover,
\[
\eu^{\iu 2\pi \theta(t)}=\varphi(t)       
\]
for every $t\in [0,3]$. Then, one can choose $F(k_1, k_2 )=\eu^{\iu 2\pi \omega(k_1,k_2 )}$, where

\[
\omega(k_1,k_2 ) :=
\begin{cases}
-2 k_2  \,\,\theta\(\frac{k_2  -2 k_1+1/2}{4 k_2 }\)&\text{if $k_2  \leq -|2 k_1 -1/2|$},\\ \\
(4 k_1 -1)\, \theta\(1+ \frac{k_2 }{4 k_1 -1}\)&\text{if $-|2 k_1 -1/2|\leq k_2 \leq  |2 k_1 -1/2|$, with $k_1 \geq 1/4$},\\ \\
2 k_2  \,\,\theta\( 2 - \frac{k_2  +  2 k_1 -1/2}{4 k_2 }\)&\text{if $k_2  \geq |2 k_1 -1/2|$},\\ \\
(-4 k_1+1)\, \theta\(\frac{5}{2} - \frac{k_2 }{-4 k_1+1}\)&\text{if $-|2 k_1 -1/2|\leq y\leq |2 k_1 -1/2|$, with $k_1 \leq 1/4$}.
\end{cases}
\]

\noindent Note that $\omega$ is continuous at $(1/4,0)$ with $\omega(1/4,0)=0$, since $\theta$ is continuous on $[0,3]$ and so there exist a $|\theta|_{\mathrm{max}}\in \mathbb{R}$ such that $|\theta(t)|\leq |\theta|_{\mathrm{max}}$ for any $t\in [0,3]$.

\medskip 

As for problem (b), while a construction of the continuous extension is possible for every $m \in \N$, here we provide the details
only for $m=2$, which is the case of interest for the $2$-bands models, as \eg  the celebrated Haldane model \cite{Haldane88}, and is such that an extension can be made completely explicit by elementary techniques. To obtain an extension for higher $m$'s, one can reduce to the case $m=2$  by recursively exploiting the fibrations  $\mathcal{SU}(m-1) \longrightarrow \mathcal{SU}(m)  \longrightarrow S^{2m-1}$.

Let  $f^{\flat} : k \mapsto {U}^{\flat}(k) \in \mathcal{SU}(2)$  be a  piecewise-smooth function. Then a  piecewise-smooth extension $F^{\flat} : \Bred\to \mathcal{SU}(2)$  is constructed as follows. First, we use the standard identification of $\mathcal{SU}(2)$ with the 3-sphere of unit norm vectors in $\mathbb{R}^4$ to look at $f^{\flat}$ as to a piecewise-smooth function $f^{\flat}  :  \partial\Bred\to S^3$.  Let $p\in S^3$ be a point not in the range
\footnote{Such a point does exists since the map $f^{\flat}$ is piecewise-smooth. Indeed,  by an argument analogous to the Sard lemma, one can show that the range of  $f^{\flat}$ is not dense in $S^3$. This is the only point in the construction where we 
need $U$ piecewise-smooth, and hence a smooth input frame $\Psi$. 
\label{foot:Sard} } 
of $f^{\flat}$, and let $\psi_p:S^3\setminus\{p\}\to \set{p}^\perp$ be the stereographic projection from $p$ to the hyperplane through the origin of $\mathbb{R}^4$ orthogonal to the vector $p$. Explicitly, this map and its inverse read 
\begin{equation} \label{Def:stereographic?}
\begin{aligned} 
\psi_p(v)  &=&  p     &-\frac{1}{\langle v- p \, | \, p\rangle } \,\, (v-p) \\
\psi_p^{-1}(w)&=&p&+\frac{2}{||w-p||^2} \,\,(w-p)\,.
\end{aligned}
\end{equation}

\noindent Second, we define a piecewise-smooth function $\varphi ^{\flat} :  [0,3]\to S^3$ by using the same formula as in \eqref{Eq:Phase}, with $f$ replaced by  $f^{\flat}$. 

\noindent Then, a piecewise-smooth extension of $f^{\flat}$ to a function $F^{\flat} : \Bred \to S^3$ is 
 given by

$$
F^{\flat} (k_1,k_2)=
\begin{cases}
\psi_{p}^{-1}\(-2k_2\,\psi_p(\varphi(\frac{k_2-2k_1+1/2}{4k_2}))\)&\text{if $k_2\leq -|2k_1-1/2|$},\\ \\
\psi_{p}^{-1}\((4k_1-1)\,\psi_p(\varphi(1+ \frac{k_2}{4k_1-1}))\)&\text{if $-|2k_1-1/2|\leq k_2\leq  |2k_1-1/2|$ and $k_1\geq 1/4$},\\ \\
\psi_{p}^{-1}\(2k_2\,\,\psi_p(\varphi( 2 - \frac{k_2+2k_1-1/2}{4k_2}))\)&\text{if $k_2\geq |2k_1-1/2|$},\\ \\
\psi_{p}^{-1}\((-4k_1+1)\,\psi_p(\varphi(\frac{5}{2} - \frac{k_2}{-4k_1+1}))\)&\text{if $-|2k_1-1/2|\leq k_2\leq |2k_1-1/2|$ and  $k_1\leq 1/4$}.
\end{cases}
$$

\medskip

\noindent Notice that $F^{\flat} $ is continuous at $(1/4,0)$ with $F^{\flat} (1/4,0)=- p$, since $\psi_p\circ\varphi:[0,3]\to p^\perp\subseteq \mathbb{R}^4$ is continuous on $[0,3]$.  This provides an explicit piecewise-smooth 
extension of $f^{\flat}$ for the case $m=2$.  
\hfill $\lozenge$
\end{rmk}

\newpage

\subsection{Interlude: abstracting from the $1$- and $2$-dimensional case}
\label{Sec:interlude} 

Abstracting from the proofs in Subsections \ref{Sec:1_dimensional} and \ref{Sec:2_dimensional}, 
we distill two Lemmas which will become the \virg{building bricks} of the higher dimensional construction. 
To streamline the statements, we denote by $\B^{(d)}$ (resp.\ $\Bred^{(d)}$) the $d$-dimensional unit cell 
(resp.\ effective unit cell) and we adhere to the following convention:
\begin{align}  \label{?Cell conventions}
\quad \qquad & \B^{(0)} \simeq \set{0},   \qquad \qquad  \Bred^{(1)} \simeq [0, \half],  \nonumber \\[-2mm]
\\[-2mm]
\quad \qquad & \Bred^{(d+1)} \simeq  \Bred^{(1)}  \times \B^{(d)}  = \big\{ (k_1, \underbrace{k_2, \ldots, k_{d+1}}_{k_{\perp}}) : \,\, k_1 \in [0, \half] ,  \,\, k_{\perp} \in \B^{(d)} \big\}  \qquad  \text{ for } d \geq 1. \nonumber
\end{align}
We also refer to the following statement as the \emph{\bf $d$-dimensional problem}: \newline
\emph{Given a continuous Bloch frame  $\Psi : \Bred^{(d)} \to \Hi^m$,  construct a continuous Bloch frame  
$
\cell{\Phi} : \Bred^{(d)} \to \Hi^m
$ 
which, via \eqref{Eq:Symmetric extension}, continuously  extends to a global continuous  symmetric Bloch 
frame $\Phi : \R^d \to \Hi^m$}. \newline
In other words, $\cell{\Phi}$ is defined only on the effective unit cell $\Bred^{(d)}$, but satisfies all the relations on $\partial \Bred^{(d)}$ (involving vertices, edges, faces, \ldots) which allow for a continuous symmetric extension to the whole $\R^d$. Hereafter, we will not further emphasize the fact that all the functions appearing in the construction are piecewise-smooth 
whenever $\Psi$ is smooth, since this fact will be used only in Remark \ref{Rem:ExplicitFilling 3d}. 

\smallskip

Notice that Subsection \ref{Sec:VertexConditions} already contains a solution to the $0$-dimensional problem:  indeed,  in view of Corollary \ref{Cor:U_factorization},  for every $\la \in \Lambda$   there exists a Bloch frame, defined on the point $k_{\la}$, satisfying the vertex condition \eqref{VertexCondition}, thus providing a solution to the $0$-dimensional problem in  $\B^{(0)} \simeq \set{k_{\la}} \simeq \set{0}$. 

\smallskip

A second look to Subsection \ref{Sec:1_dimensional} shows that it contains a solution to the $1$-dimensional problem, given a solution to the $0$-dimensional problem.  Indeed, one extracts from the construction the following Lemma. 

\begin{lemma}[Macro 1] \label{Lem:Macro1}
Let $\Phi\uno: \B^{(0)} \simeq \set{0}  \to \Hi^m$ be a Bloch frame satisfying 
\begin{equation} \label{Eq:0_dimensional}
\Phi\uno(0) =  \Theta \, \Phi\uno(0).
\end{equation}
Then one constructs a continuous Bloch frame $\Phi \due: \Bred^{(1)} \simeq [0, \half] \to \Hi^m$ such that
\begin{equation} \label{Eq:Sym_v1}
\begin{cases}
\Phi\due(0)  = \Phi\uno(0)  \\
\Phi \due(\half)  =  \tau_{e_2} \Theta \, \Phi\due (\half).    
\end{cases} \end{equation}
\end{lemma}

In view of \eqref{Eq:Sym_v1},  ${\Phi}\due$  continuously extends, via \eqref{Eq:Symmetric extension}, to a global continuous symmetric Bloch frame, thus providing a solution to the $1$-dimensional problem.

\smallskip

Analogously, from the construction in  Subsection \ref{Sec:2_dimensional}  we distill a general procedure. For convenience, 
we relabel the edges of $\Bred^{(2)}$ as follows: 
\begin{eqnarray} \label{Def:Edges}
E_{j,0} &=& \set{k = \sum_i k_i e_i \in \Bred^{(2)} :  k_j = 0} \\              
E_{j, \pm} &=& \set{k = \sum_i k_i e_i  \in \Bred^{(2)} :  k_j = \pm \half}    
\end{eqnarray}
From the construction in Subsection  \ref{Sec:2_dimensional}, based on Lemma \ref{Lem:Xmatrix},  one easily deduces the following result. 
\begin{lemma}[Macro 2]\label{Lem:Macro2} 
Assume that   $\Phi_S:  S  \to \Hi^m$,  with $S = E_{1,0} \cup E_{2,-} \cup E_{2,+} \subset \partial \Bred^{(2)}$,  
is continuous and satisfies the following edge symmetries: 
\begin{equation}  \label{Eq: S conditions}
\begin{cases}
\Phi_S(t_{e_2}(k)) =  \tau_{e_2} \Phi_S(k)                      & \qquad \text{ for } k \in E_{2,-} \\
\Phi_S(t_{e_2}^{-1}(k)) =  \tau^{-1}_{e_2} \Phi_S(k)       & \qquad  \text{ for } k \in E_{2,+} \\
\Phi_S(c(k)) =  \Theta \, \Phi_S(k)                                      & \qquad  \text{ for } k \in E_{1,0}.
\end{cases}
\end{equation}
Then one constructs a continuous Bloch frame $\cell{\Phi}: \Bred^{(2)} \to \Fr(m, \Hi)$ such that
\begin{equation} \label{Eq: S complement}
\begin{cases}
\cell{\Phi}(k) = \Phi_S(k)                                                         & \text{ for } k \in S \\
\cell{\Phi}(t_{e_1} c(k) ) = \tau_{e_1} \Theta \, \cell{\Phi}(k)      & \text{ for }  k \in E_{1,+}.
\end{cases}
\end{equation}
\end{lemma}
To obtain \eqref{Eq: S complement}  we implicitly used property \ref{item:iii} in Lemma \ref{Lem:Xmatrix}, 
which guarantees that it is possible to obtain the frame $\cell{\Phi}$ by acting only on the edge $E_{1,+} = E_3 \cup E_4$.   
Notice that, in view of \eqref{Eq: S complement},  $\cell{\Phi}$  continuously extends, via \eqref{Eq:Symmetric extension}, to a global continuous symmetric Bloch frame.  Therefore, a solution to the $2$-dimensional problem can always be constructed, whenever a continuous Bloch frame on the $1$-dimensional set $S$, satisfying the edge symmetries \eqref{Eq: S conditions}, is provided. 

\medskip

The previous Lemmas \ref{Lem:Macro1} and \ref{Lem:Macro2} will yield a constructive and conceptually clear solution to the $3$-dimensional problem, and a characterization of the obstruction to the solution to the $4$-dimensional problem.   


\subsection{Construction in the $3$-dimensional case}
\label{Sec:3_dimensional}

The faces of $\partial \Bred^{(3)}$ are labelled according to the following convention :
for  $j \in \set{1, 2, 3}$ we set 
\begin{eqnarray} \label{Def:Faces}
F_{j,0} &=& \set{k = \sum_{i=1}^3 k_i e_i :  k_j = 0} \\              
F_{j, \pm} &=& \set{k = \sum_{i=1}^3 k_i e_i :  k_j = \pm \half}    
\end{eqnarray}
Notice that two faces of $\partial \Bred^{(3)}$, namely $F_{1,0}$ and $F_{1,+}$ are identifiable with a $2$-dimensional unit cell $\B^{(2)}$,  while the remaining four faces, namely  $F_{2, \pm}$ and $F_{3,\pm}$, are identifiable with a   $2$-dimensional effective  unit cell $\Bred^{(2)}$. 


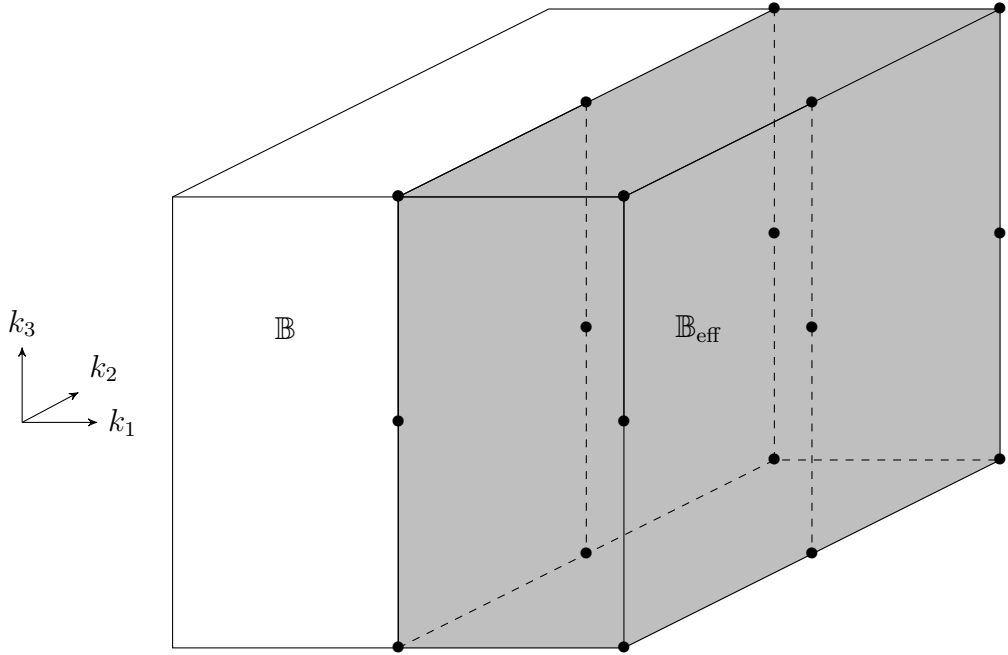
\begin{figure}[ht]
\centering
\scalebox{1}{%
\begin{tikzpicture}[>=stealth']
\filldraw [lightgray] (0,3) -- (0,-3) -- (3,-3) -- (8,-0.5) -- (8,5.5) -- (5,5.5) -- (0,3);
\draw (0,3) -- (0,-3) -- (-3,-3) -- (-3,3) -- (0,3);
\draw (-3,3) -- (2,5.5) -- (5,5.5);
\draw (0,3) node {$\bullet$}
	   -- (0,-3) node {$\bullet$}
	   -- (3,-3) node {$\bullet$}
	   -- (3,3) node {$\bullet$}
	   -- (0,3) node {$\bullet$}
	   -- (5,5.5) node {$\bullet$}
	   -- (8,5.5) node {$\bullet$}
	   -- (3,3) node {$\bullet$};
\draw (3,-3)
	   -- (8,-0.5) node {$\bullet$}
	   -- (8,5.5);
\draw [dashed] (0,-3)
				 -- (5,-0.5) node {$\bullet$}
				 -- (8,-0.5);
\draw [dashed] (5,-0.5) -- (5,5.5);
\draw [dashed] (5.5,4.25) node {$\bullet$} -- (5.5,1.25) node {$\bullet$};
\draw [dashed] (5.5,1.25) -- (5.5,-1.75) node {$\bullet$};
\draw [dashed] (2.5,4.25) node {$\bullet$} 
						 -- (2.5,1.25) node {$\bullet$};
\draw [dashed] (2.5,1.25) -- (2.5,-1.75) node {$\bullet$};
\draw (2.5,4.25) -- (0,3) -- (3,3) -- (5.5,4.25);
\draw (0,3) -- (0,0) node {$\bullet$};
\draw (3,3) -- (3,0) node {$\bullet$};
\draw (8,2.5) node {$\bullet$}
	  (5,2.5) node {$\bullet$}
	  (4,1.25) node {$\Bred$}
	   (-1.5,1.25) node {$\B$};
\draw [->] (-5,0) -- (-4,0) node [anchor=west] {$k_1$};
\draw [->] (-5,0) -- (-4.25,.4) node [anchor=south west] {$k_2$};
\draw [->] (-5,0) -- (-5,1) node [anchor=south] {$k_3$};
\end{tikzpicture}
}
\caption{The $3$-dimensional (effective) unit cell.}
\label{fig:Bred3}
\end{figure}


We assume as given a continuous Bloch frame $\Psi: \Bred^{(3)} \to \Hi^m$ (the \emph{input frame}) which does not satisfy 
any particular symmetry on the boundary $\partial \Bred^{(3)}$.  Since $F_{1,0} \simeq \B^{(2)}$, in view of the construction in Subsection \ref{Sec:2_dimensional} we can assume that $\Psi$ has been already modified to obtain a continuous Bloch frame 
$\Phi\uno : F_{1,0} \to \Hi^m$, $\Phi\uno = \Psi \act U\uno$, which satisfies the edge symmetries \eqref{EdgeSymmetries} on $F_{1,0}$. 

For convenience, we decompose the constructive algorithm into few steps: 

\textsl{$\bullet$ Step 1. Extend to an edge by Macro 1.} Choose a vertex of  $\partial \Bred^{(3)}$ contained in $F_{1,+}$, and let $v_0$ be the corresponding vertex on $F_{1,0}$. For the sake of concreteness, we choose $v_{*} = (\half,\half,\half)$, so that   $v_{0} = (0,\half,\half)$.  Then Lemma \ref{Lem:Macro1} (Macro 1) yields the existence of a continuous Bloch frame $ \Phi\due :  [v_0, v_*] \simeq  \Bred^{(1)}  \to \Hi^m$ such that 
$\Phi\due(v_0) = \Phi\uno(v_0)$ and $\Phi\due(v_*)$ satisfies the vertex condition \eqref{VertexCondition} at $v_*$. 

\textsl{$\bullet$ Step  2. Extend by $\tau$-equivariance.} By imposing property $\mathrm{(F_2)}$,  $\Phi\due$ naturally extends to the edges $t^{-1}_{e_2}([v_0,v_*])$ and 
$t^{-1}_{e_3}([v_0,v_*])$. Since $\Phi\uno$ is $\tau$-equivariant on $F_{1,0} \simeq \B^{(2)}$, one has that  
$\Phi\due(t^{-1}_{e_j}(v_0)) = \Phi\uno(t^{-1}_{e_j}(v_0))$ for $j \in \set{2,3}$.  
In view of that, we obtain a continuous Bloch frame  by setting
\begin{equation} 
\Phi\tre(k) :=  \begin{cases}
\Phi\uno(k)        & \text{ for } k \in F_{1,0}, \\
\Phi\due(k)        & \text{ for } k \in t_{\lambda}([v_0,v_*])  \text{ for } \lambda \in \set{0, -e_2, -e_3}. 
\end{cases}
\end{equation}

\textsl{$\bullet$ Step 3. Extend to small faces by Macro 2.} Notice that $\Phi\tre$ restricted to $F_{3,+}$ (resp.\ $F_{2,+}$) is defined and continuous on a set $S_{3,+}$ (resp.\  $S_{2,+}$) which has the same structure as the set $S$ appearing in Lemma \ref{Lem:Macro2} (Macro 2), and there satisfies the relations analogous to  \eqref{Eq: S conditions}. Then $\Phi\tre$ continuously extends to the whole $F_{3,+}$ (resp.\ $F_{2,+}$) and the extension  satisfies the relation analogous to \eqref{Eq: S complement} on the edge $\partial F_{j,+} \setminus S_{j,+}$  for $j =3$ (resp.\ $j =2$). 


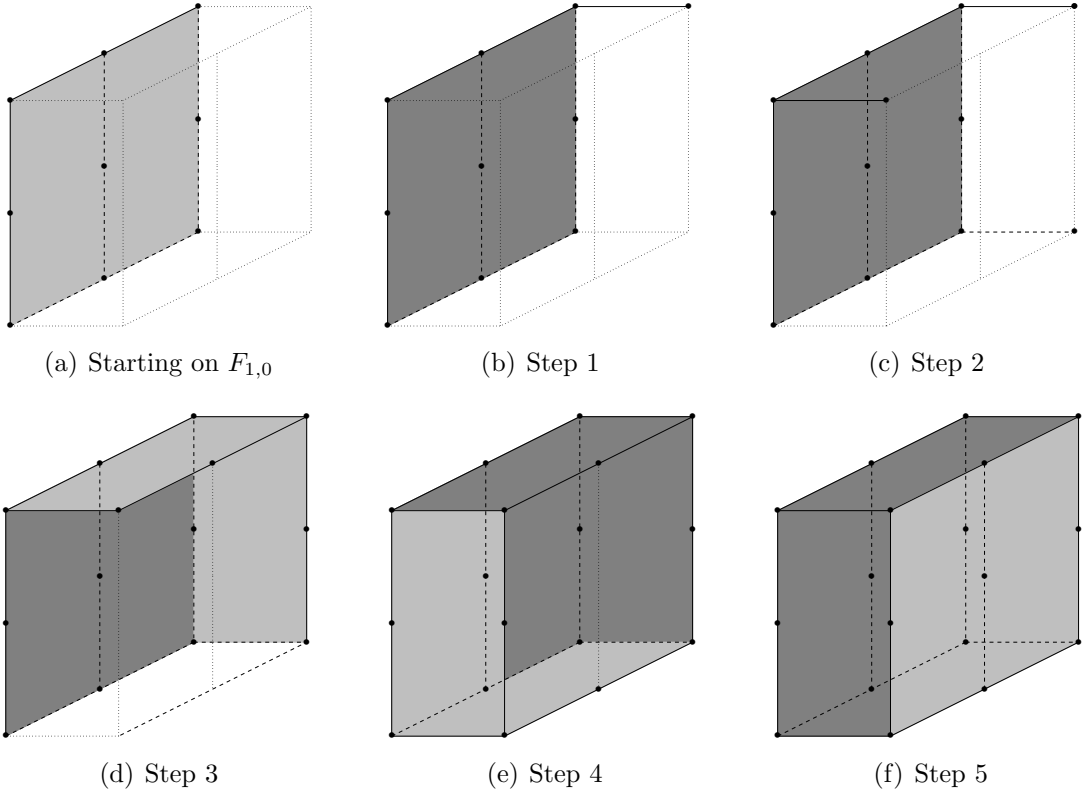
\begin{figure}[ht]
\centering
\subfigure[Starting on $F_{1,0}$]{\scalebox{.5}{%
\begin{tikzpicture}
\filldraw [lightgray] (0,3) -- (0,-3) -- (5,-0.5) -- (5,5.5) -- (0,3);
\draw (5,5.5) -- (0,3) node {$\bullet$} -- (0,-3) node {$\bullet$}; 
\draw [dashed] (0,-3) -- (5,-0.5) node {$\bullet$} -- (5,5.5) node {$\bullet$}; 
\draw [dotted] (3,-3)
	   -- (3,3)
	   -- (8,5.5)
	   -- (8,-0.5)
	   -- (3,-3) -- (0,-3);
\draw [dotted] (0,3) -- (3,3);
\draw [dotted] (5,5.5) -- (8,5.5);
\draw [dotted] (5,-0.5) -- (8,-0.5);
\draw [dotted] (5.5,4.25)
			     -- (5.5,1.25);
\draw [dotted] (5.5,1.25) -- (5.5,-1.75);
\draw [dashed] (2.5,4.25) node {$\bullet$} 
	   -- (2.5,1.25) node {$\bullet$}
	   -- (2.5,-1.75) node {$\bullet$};
\draw  (0,0) node {$\bullet$}
		(5,2.5)  node {$\bullet$};
\end{tikzpicture}
}}\qquad
\subfigure[Step 1]{\scalebox{.5}{%
\begin{tikzpicture}
\filldraw [gray] (0,3) -- (0,-3) -- (5,-0.5) -- (5,5.5) -- (0,3);
\draw (5,5.5) -- (0,3) node {$\bullet$} -- (0,-3) node {$\bullet$}; 
\draw [dashed] (0,-3) -- (5,-0.5) node {$\bullet$} -- (5,5.5) node {$\bullet$}; 
\draw [dotted] (3,-3)
	   -- (3,3)
	   -- (8,5.5) node {$\bullet$}
	   -- (8,-0.5)
	   -- (3,-3) -- (0,-3);
\draw [dotted] (0,3) -- (3,3);
\draw (5,5.5) -- (8,5.5);
\draw [dotted] (5,-0.5) -- (8,-0.5);
\draw [dotted] (5.5,4.25)
			     -- (5.5,1.25);
\draw [dotted] (5.5,1.25) -- (5.5,-1.75);
\draw [dashed] (2.5,4.25) node {$\bullet$} 
	   -- (2.5,1.25) node {$\bullet$}
	   -- (2.5,-1.75) node {$\bullet$};
\draw  (0,0) node {$\bullet$}
		(5,2.5)  node {$\bullet$};
\end{tikzpicture}
}}\qquad
\subfigure[Step 2]{\scalebox{.5}{%
\begin{tikzpicture}
\filldraw [gray] (0,3) -- (0,-3) -- (5,-0.5) -- (5,5.5) -- (0,3);
\draw (5,5.5) -- (0,3) node {$\bullet$} -- (0,-3) node {$\bullet$}; 
\draw [dashed] (0,-3) -- (5,-0.5) node {$\bullet$} -- (5,5.5) node {$\bullet$}; 
\draw [dotted] (3,-3)
	   -- (3,3)
	   -- (8,5.5) node {$\bullet$}
	   -- (8,-0.5) node {$\bullet$};
\draw [dotted] (3,-3) -- (0,-3);
\draw (0,3) -- (3,3) node {$\bullet$};
\draw (5,5.5) -- (8,5.5) node {$\bullet$};
\draw [dashed] (5,-0.5) -- (8,-0.5); 
\draw [dotted] (8,-0.5) -- (3,-3);
\draw [dotted] (5.5,4.25)
			     -- (5.5,1.25);
\draw [dotted] (5.5,1.25) -- (5.5,-1.75);
\draw [dashed] (2.5,4.25) node {$\bullet$} 
	   -- (2.5,1.25) node {$\bullet$}
	   -- (2.5,-1.75) node {$\bullet$};
\draw  (0,0) node {$\bullet$}
		(5,2.5)  node {$\bullet$};
\end{tikzpicture}
}}\\
\subfigure[Step 3]{\scalebox{.5}{%
\begin{tikzpicture}
\filldraw [gray] (0,3) -- (0,-3) -- (5,-0.5) -- (5,5.5) -- (0,3);
\filldraw [lightgray] (0,3) -- (3,3) -- (8,5.5) -- (5,5.5) -- (0,3);
\filldraw [lightgray] (8,5.5) -- (5,5.5) -- (5,-0.5) -- (8,-0.5) -- (8,5.5);
\draw (5,5.5) -- (0,3) node {$\bullet$} -- (0,-3) node {$\bullet$}; 
\draw [dashed] (0,-3) -- (5,-0.5) node {$\bullet$} -- (5,5.5) node {$\bullet$}; 
\draw [dotted] (3,-3)
	   -- (3,3); 
\draw (3,3)
	   -- (8,5.5);
\draw (8,5.5) -- (8,-0.5) node {$\bullet$};
\draw [dotted] (3,-3) -- (0,-3);
\draw (0,3) -- (3,3) node {$\bullet$};
\draw (5,5.5) -- (8,5.5) node {$\bullet$};
\draw [dashed] (5,-0.5) -- (8,-0.5) -- (3,-3);
\draw [dotted] (5.5,4.25) node {$\bullet$} 
			     -- (5.5,1.25);
\draw [dotted] (5.5,1.25) -- (5.5,-1.75);
\draw [dashed] (2.5,4.25) node {$\bullet$} 
	   -- (2.5,1.25) node {$\bullet$}
	   -- (2.5,-1.75) node {$\bullet$};
\draw  (0,0) node {$\bullet$}
       (8,2.5) node {$\bullet$}
		(5,2.5)  node {$\bullet$};
\end{tikzpicture}
}}\qquad
\subfigure[Step 4]{\scalebox{.5}{%
\begin{tikzpicture}
\filldraw [gray] (5,-0.5) -- (5,5.5) -- (8,5.5) -- (8,-0.5) -- (5,-0.5);
\filldraw [lightgray] (5,-0.5) -- (8,-0.5) -- (3,-3) -- (0,-3) -- (5,-0.5);
\filldraw [gray] (0,3) -- (0,-3) -- (5,-0.5) -- (5,5.5) -- (0,3);
\filldraw [gray] (0,3) -- (3,3) -- (8,5.5) -- (5,5.5) -- (0,3);
\filldraw [lightgray] (0,3) -- (3,3) -- (3,-3) -- (0,-3) -- (0,3);
\draw (5,5.5) -- (0,3) node {$\bullet$} -- (0,-3) node {$\bullet$}; 
\draw [dashed] (0,-3) -- (5,-0.5) node {$\bullet$} -- (5,5.5) node {$\bullet$}; 
\draw (3,-3)
	   -- (3,3)
	   -- (8,5.5);
\draw (8,5.5) -- (8,-0.5) node {$\bullet$};
\draw (3,-3) node {$\bullet$} -- (0,-3);
\draw (0,3) -- (3,3) node {$\bullet$};
\draw (5,5.5) -- (8,5.5) node {$\bullet$};
\draw [dashed] (5,-0.5) -- (8,-0.5);
\draw (8,-0.5) -- (3,-3);
\draw [dotted] (5.5,4.25) node {$\bullet$} 
			     -- (5.5,1.25)
				 -- (5.5,-1.75) node {$\bullet$};
\draw [dashed] (2.5,4.25) node {$\bullet$} 
	   -- (2.5,1.25) node {$\bullet$}
	   -- (2.5,-1.75) node {$\bullet$};
\draw  (0,0) node {$\bullet$}
       (3,0) node {$\bullet$}
		(5,2.5) node {$\bullet$}
		(8,2.5) node {$\bullet$};
\end{tikzpicture}
}}\qquad
\subfigure[Step 5]{\scalebox{.5}{%
\begin{tikzpicture}
\filldraw [gray] (0,3) -- (3,3) -- (8,5.5) -- (5,5.5) -- (0,3);
\filldraw [gray] (0,3) -- (3,3) -- (3,-3) -- (0,-3) -- (0,3);
\filldraw [lightgray] (3,-3) -- (8,-0.5) -- (8,5.5) -- (3,3) -- (3,-3);
\draw (5,5.5) -- (0,3) node {$\bullet$} -- (0,-3) node {$\bullet$}; 
\draw [dashed] (0,-3) -- (5,-0.5) node {$\bullet$} -- (5,5.5) node {$\bullet$}; 
\draw (3,-3)
	   -- (3,3)
	   -- (8,5.5);
\draw (8,5.5) -- (8,-0.5) node {$\bullet$};
\draw (3,-3) node {$\bullet$} -- (0,-3);
\draw (0,3) -- (3,3) node {$\bullet$};
\draw (5,5.5) -- (8,5.5) node {$\bullet$};
\draw [dashed] (5,-0.5) -- (8,-0.5);
\draw (8,-0.5) -- (3,-3);
\draw [dashed] (5.5,4.25) node {$\bullet$} 
			     -- (5.5,1.25) node {$\bullet$}
				 -- (5.5,-1.75) node {$\bullet$};
\draw [dashed] (2.5,4.25) node {$\bullet$} 
	   -- (2.5,1.25) node {$\bullet$}
	   -- (2.5,-1.75) node {$\bullet$};
\draw  (0,0) node {$\bullet$}
       (3,0) node {$\bullet$}
		(5,2.5) node {$\bullet$}
		(8,2.5) node {$\bullet$};
\end{tikzpicture}
}}
\caption{Steps in the construction.}
\label{fig:Steps}
\end{figure}

\textsl{$\bullet$ Step  4. Extend by $\tau$-equivariance.} By imposing $\tau$-equivariance (property \ref{item:F2}), $\Phi\tre$ naturally extends to the faces $F_{3,-}$ and $F_{2,-}$, thus yielding a continuous Bloch frame $\Phi\quat$ defined on the set
\footnote{According to a longstanding tradition in geometry, the choice of symbols is inspired by the German language:  $K_0$ stands for \emph{Kleiderschrank ohne T\"{u}ren}. The reason for this name will be clear in few lines. }  
\  $\partial \Bred^{(3)} \setminus F_{1,+}  =: K_{0}$.
 
\textsl{$\bullet$ Step  5. Extend symmetrically to $F_{1,+}$ by Macro 2.} When considering the face $F_{1,+}$, we first notice that the two subsets\footnote{Obviously, these subsets are \emph{die T\"{u}ren}, so they are denoted by $T_{\pm}$.} 
\begin{equation} \label{Turen}
T_{\pm} = \set{k \in F_{1,+} :  \pm \, k_2 \geq 0 }
\end{equation}       
are related by a non-trivial symmetry, since $t_{e_1} c \, (T_{\pm}) = T_{\mp}$. We construct a continuous extension of $\Phi\quat$ which is compatible with the latter symmetry.  

The restriction of $\Phi\quat$, defined on $K_0$,  to the set  $ S_{+} = \set{k \in \partial F_{1,+} :  k_2 \geq 0 }$
is continuous and satisfies symmetries analogous to \eqref{Eq: S conditions}. Then, in view of Lemma \ref{Lem:Macro2} (Macro 2), $\Phi\quat$  continuously extends to the whole $T_{+}$ and the extension satisfies the relation analogous to \eqref{Eq: S complement} on the edge $\partial T_{+} \setminus S_{+}$. We denote the extension by $\Phi\cinque$. 

To obtain a continuous symmetric Bloch frame $\widehat \Phi :  \partial \Bred^{(3)} \to \Fr(m, \Hi)$ we set
\begin{equation} 
\label{Phi cinque}
\widehat \Phi (k) :=  \begin{cases}
 \Phi\quat(k)                            & \text{ for } k \in K_0   \\
 \Phi\cinque(k)                        & \text{ for }  k \in T_{+}   \\
  \Phi\cinque(t_{e_1}c(k))       & \text{ for }   k \in T_{-}. 
\end{cases}
\end{equation}
The function $\widehat \Phi$ is continuous in view of the edge and face symmetries which have been imposed in the construction. 
  
\textsl{$\bullet$ Step 6. Extend to the interior of the effective cell.} The frame $\widehat \Phi $  and the input frame $\Psi$ are related by the equation
\begin{equation} \label{hat Phi}
\widehat \Phi (k) =   \Psi(k) \act \widehat U (k)  \qquad \qquad k \in \partial \Bred^{(3)}, 
\end{equation}
which yields a continuous map $ \widehat U :  \partial \Bred^{(3)} \to \U(\C^m)$. 

We show that such a map extends to a continuous map $ \cell{U} :  \Bred^{(3)} \to \U(\C^m)$. Indeed, a continuous function $f$ from  $\partial \Bred^{(3)} \approx S^2$ to the topological space $X$  can be continuously extended to  
$\Bred^{(3)} \approx D^3$  if and only if its homotopy class $[f]$ is the trivial element of the group $\pi_{2}(X)$.  
In our case,  since $\pi_2(\U(\C^m)) = \set{0}$ for every $m \in \N$, there is no obstruction to the continuous extension of the map  $\widehat U$.   Moreover, the extension can be explicitly constructed, as detailed in Remark 
\ref{Rem:ExplicitFilling 3d}. \newline 
Equipped with such a continuous extension, we obtain a continuous symmetric Bloch frame 
 $ \cell{\Phi} :  \Bred^{(3)} \to \Hi^m$ by setting
\begin{equation} \label{cell Phi}
\cell{\Phi} (k) =   \Psi(k) \act \cell{U}(k)   \qquad \qquad k \in \Bred^{(3)}. 
\end{equation}  

\textsl{$\bullet$ Step 7. Use the smoothing procedure.} By using \eqref{Eq:Symmetric extension}, $\cell{\Phi}$ extends to 
a  global continuous symmetric Bloch frame.  Then the symmetry-preserving smoothing procedure (Proposition 
\ref{Lem:SmoothingProcedure}) yields a global smooth symmetric Bloch frame.

This concludes the proof of Theorem \ref{Thm:SmoothBlochFrames} for $d=3$.

\begin{rmk}[Explicit extension to the whole effective cell, $d=3$] \label{Rem:ExplicitFilling 3d}
As in the $2$-dimensional case (Remark \ref{Rem:ExplicitFilling 2d}), we notice that the extension 
of  the \emph{piecewise-smooth}
\footnote{The map $\widehat U$ is actually piecewise-smooth, whenever the input frame $\Psi$ is smooth. Although this fact was not emphasized at every step of the $3$-dimensional construction, as we did instead in Section \ref{Sec:2_dimensional}, the reader can easily check it. 
}\  
function $ \widehat U :  \partial \Bred^{(3)} \to \U(\C^m)$ to  $\Bred^{(3)}$ is completely explicit. 
The problem is again reduced to the following two subproblems, namely to 
construct a continuous extension from $\partial \Bred^{(3)}$ to $\Bred^{(3)}$ of: 
\begin{enumerate}[label=$(\mathrm{\alph*})$,ref=$(\mathrm{\alph*})$]
\item \label{item:sp i} a map $f : k \mapsto \det U(k) \in \U(1)$, and 
\item \label{item:sp ii} a map  $f^{\flat} :k \mapsto {U}^{\flat}(k) \in \mathcal{SU}(m)$.
\end{enumerate}

As for subproblem \ref{item:sp i},  given a  degree-zero  piecewise-smooth function $f:\partial \Bred^{(3)} \to \U(1)$,  
a  piecewise-smooth extension $F: \Bred^{(3)} \to \U(1)$ is constructed as follows. 
Consider the region $\mathbb{D}\subseteq\mathbb{R}^2$ depicted below
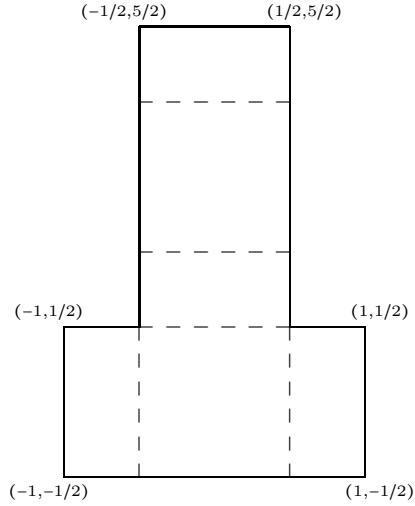
\[
\begin{xy}
,(-10,0);(30,0)**\dir{-}
,(30,0);(30,20)**\dir{-}
,(30,20);(20,20)**\dir{-}
,(20,20);(20,60)**\dir{-}
,(20,60);(0,60)**\dir{-}
,(0,60);(0,20)**\dir{-}
,(0,20);(-10,20)**\dir{-}
,(-10,20);(-10,0)**\dir{-}
,(0,20);(20,20)**\dir{--}
,(0,20);(0,0)**\dir{--}
,(20,20);(20,0)**\dir{--}
,(0,30);(20,30)**\dir{--}
,(0,50);(20,50)**\dir{--}
,(-12,-2)*{\scriptscriptstyle{(-1,-1/2)}}
,(32,-2)*{\scriptscriptstyle{(1,-1/2)}},(-12,22)*{\scriptscriptstyle{(-1,1/2)}}
,(32,22)*{\scriptscriptstyle{(1,1/2)}}
,(-2,62)*{\scriptscriptstyle{(-1/2,5/2)}}
,(22,62)*{\scriptscriptstyle{(1/2,5/2)}}
\end{xy}
\]
and let $\varphi:\mathbb{D}\to \U(1)$ be the piecewise-smooth function defined by

\begin{equation}\label{Eq: Phase_3d}
\varphi(s,t) :=
\begin{cases}
f(-s-\frac{1}{2},-\frac{1}{2}, t) & \text{if $(s,t)\in[-1,-\frac{1}{2}]\times [-\frac{1}{2},\frac{1}{2}]$}\\
f(0,s, t)& \text{if $(s,t)\in[-\frac{1}{2},\frac{1}{2}]\times [-\frac{1}{2},\frac{1}{2}]$}\\
f(s-\frac{1}{2}, \frac{1}{2}, t)& \text{if $(s,t)\in[\frac{1}{2},1]\times [-\frac{1}{2},\frac{1}{2}]$}\\
f(t-\frac{1}{2}, s, \frac{1}{2})& \text{if $(s,t)\in[-\frac{1}{2},\frac{1}{2}]\times [\frac{1}{2},1]$}\\
f(\frac{1}{2},s, -t + \frac{3}{2})& \text{if $(s,t)\in[-\frac{1}{2},\frac{1}{2}]\times [1,2]$}\\
f(-t + \frac{5}{2},s, -\frac{1}{2})& \text{if $(s,t)\in[-\frac{1}{2},\frac{1}{2}]\times [2,\frac{5}{2}]$}
\end{cases}
\end{equation}

\noindent Choose $\theta_0\in \mathbb{R}$ be such that $f(0,0,0)= \eu^{2\pi \iu \theta_0}$ and set
\[
\theta(s,t)=\theta_0+\frac{1}{2\pi i}\int_0^1\varphi(\lambda s,\lambda t)^{-1}\left(s\varphi_s(\lambda s,\lambda t)+t\varphi_t(\lambda s,\lambda t)\right)d\lambda,
\]
where $\varphi_s$ and $\varphi_t$ denote the partial derivatives of $\varphi$ with respect to $s$ and $t$, respectively. One has
\[
\eu^{2\pi \iu \theta(s,t)}=\varphi(s,t)
\]
for any $(s,t)\in \mathbb{D}$. Moreover, 
\[
\begin{cases}
\theta(s,-\frac{1}{2})=\theta(s,\frac{5}{2}) &\text{for $-\frac{1}{2}\leq s\leq\frac{1}{2}$}\\
\theta(s,-\frac{1}{2})=\theta(-\frac{1}{2},3+s) &\text{for $-1\leq s\leq-\frac{1}{2}$}\\
\theta(s,-\frac{1}{2})=\theta(\frac{1}{2},3-s) &\text{for $\frac{1}{2}\leq s\leq1$}\\
\theta(s,\frac{1}{2})=\theta(-\frac{1}{2},\frac{1}{2}-s) &\text{for $-1\leq s\leq-\frac{1}{2}$}\\
\theta(s,\frac{1}{2})=\theta(\frac{1}{2},\frac{1}{2} + s) &\text{for $\frac{1}{2}\leq s\leq 1$}\\
\theta(-1,t)=\theta(-\frac{1}{2},\frac{3}{2}-t) &\text{for $- \frac{1}{2}\leq t\leq \frac{1}{2}$}\\
\theta(1,t)=\theta(\frac{1}{2},\frac{3}{2}-t) &\text{for $\frac{1}{2}\leq t\leq \frac{1}{2}$}\\
\end{cases}
\]
so that $\theta$ actually lifts $f$ to a piecewise-smooth function $\theta:\partial\Bred^{(3)}\to \mathbb{R}$.
Then, one can choose $F(k_1,k_2,k_3)= \eu^{2\pi \iu \, \omega(k_1,k_2,k_3)}$, where
\[
\omega(k_1,k_2,k_3)=
\begin{cases}
(-4k_1+1) \, \theta\(-\frac{k_3}{4k_1-1},-\frac{k_2}{4k_1-1}\)&\text{if $0\leq k_1\leq \min\{|\frac{1}{2}|k_2|-\frac{1}{4}|,|\frac{1}{2}|k_3|-\frac{1}{4}|\}$},\\ \\
2k_3\,\, \theta\(\frac{1}{2}+\frac{k_3+2k_1-1/2}{4k_3},\frac{k_2}{2k_3}\)&\text{if $\max\{|2k_1-\frac{1}{2}|,|k_2|\}\leq k_3\leq \frac{1}{2}$},\\ \\
2k_2\,\, \theta\(\frac{k_3}{2k_2},\frac{1}{2}+\frac{k_2+2k_1-1/2}{4k_2}\)&\text{if $\max\{|2k_1-\frac{1}{2}|,|k_3|\}\leq k_2\leq \frac{1}{2}$},\\ \\
(4k_1-1)\, \theta\(\frac{k_3}{4k_1-1},\frac{3}{2}-\frac{k_2}{4k_1-1}\)&\text{if $ \max\{|\frac{1}{2}|k_2|+\frac{1}{4}|,|\frac{1}{2}|k_3|+\frac{1}{4}|\}\leq k_1\leq \frac{1}{2}$},\\ \\
-2k_3\,\, \theta\(-\frac{1}{2}-\frac{k_3-2k_1+1/2}{4k_3},-\frac{k_2}{2k_3}\)&\text{if $-\frac{1}{2}\leq k_3\leq \min\{-|2k_1-\frac{1}{2}|,-|k_2|\}$},\\ \\
-2k_2\,\, \theta\(-\frac{k_3}{2k_2},\frac{5}{2}-\frac{k_2-2k_1+1/2}{4k_2}\)&\text{if $-\frac{1}{2}\leq k_2\leq \min\{-|2k_1-\frac{1}{2}|,-|k_3|\}$}.\\ \\
\end{cases}
\]
Note that $\omega$ is continuous at $(1/4,0,0)$ with $\omega(1/4,0,0)=0$, since $\theta$ is continuous on the compact domain $\mathbb{D}$.


As mentioned in Remark \ref{Rem:ExplicitFilling 2d}, the solution to  subproblem \ref{item:sp ii} can be obtained by recursive reduction of the rank $m$, up to $m=2$.   To construct the extension in the case $m=2$,  we identify $\mathcal{SU}(2)$ with $S^3$ and use the stereographic projection \eqref{Def:stereographic?}, with respect to a point $p \in S^3$ which is not included in the range
\footnote{ This point does exist since the map  $f^{\flat}$ is piecewise-smooth, as argued in footnote \!\!\! \ref{foot:Sard}. 
}  
of  the map $f^{\flat}$.  By using Equation \eqref{Eq: Phase_3d}, with $f$ replaced by $f^{\flat}$,  one defines a piecewise-smooth function $\varphi:\mathbb{D}\to S^3$. Then, a piecewise-smooth extension of $f^{\flat}$ to a function $F^{\flat} : \Bred^{(3)}\to S^3$ is explicitly  
given by
\[
F(k_1,k_2,k_3)=
\begin{cases}
\psi_{p}^{-1}\((-4k_1+1)\psi_p\(\varphi(-\frac{k_3}{4k_1-1},-\frac{k_2}{4k_1-1})\)\)&\text{if $0\leq k_1\leq \min\{|\frac{1}{2}|k_2|-\frac{1}{4}|,|\frac{1}{2}|k_3|-\frac{1}{4}|\}$},\\ \\
\psi_{p}^{-1}\(2k_3\,\psi_p\(\varphi(\frac{1}{2}+\frac{k_3+2k_1-1/2}{4k_3},\frac{k_2}{2k_3})\)\)&\text{if $\max\{|2k_1-\frac{1}{2}|,|k_2|\}\leq k_3\leq \frac{1}{2}$},\\ \\
\psi_{p}^{-1}\(2k_2\,\psi_p\(\varphi(\frac{k_3}{2k_2},\frac{1}{2}+\frac{k_2+2k_1-1/2}{4k_2})\)\)&\text{if $\max\{|2k_1-\frac{1}{2}|,|k_3|\}\leq k_2\leq \frac{1}{2}$},\\ \\
\psi_{p}^{-1}\((4k_1-1)\psi_p\(\varphi(\frac{k_3}{4k_1-1},\frac{3}{2}-\frac{k_2}{4k_1-1})\)\)&\text{if $ \max\{|\frac{1}{2}|k_2|+\frac{1}{4}|,|\frac{1}{2}|k_3|+\frac{1}{4}|\}\leq k_1\leq \frac{1}{2}$},\\ \\
\psi_{p}^{-1}\(-2k_3\,\psi_p\(\varphi(-\frac{1}{2}-\frac{k_3-2k_1+1/2}{4k_3},-\frac{k_2}{2k_3})\)\)&\text{if $-\frac{1}{2}\leq k_3\leq \min\{-|2k_1-\frac{1}{2}|,-|k_2|\}$},\\ \\
\psi_{p}^{-1}\(-2k_2,\psi_p\(\varphi(-\frac{k_3}{2k_2},\frac{5}{2}-\frac{k_2-2k_1+1/2}{4k_2})\)\)&\text{if $-\frac{1}{2}\leq k_2\leq \min\{-|2k_1-\frac{1}{2}|,-|k_3|\}$}.\\ \\
\end{cases}
\]
Note that $F$ is continuous at $(1/4,0,0)$ with $F(1/4,0,0)=-p$, since $\psi_p\circ\varphi: \mathbb{D} \to p^\perp\subseteq \mathbb{R}^4$ is continuous on the compact domain $\mathbb{D}$. 
The map above provides an explicit continuous extension to $\Bred^{(3)}$ for $m =2$. \hfill $\lozenge$
\end{rmk}


\subsection{A glimpse to  the higher-dimensional cases}
\label{Sec:h_dimensional}

The fundamental \virg{building bricks} used to solve the $3$-dimensional problem (Lemma \ref{Lem:Macro1} and \ref{Lem:Macro2})
can be used to approach the higher-dimensional problems. However, additional topological obstruction might appear, related to the fact that the $k$-th homotopy group $\pi_k \( \U(\C^m)\)$, for $k \geq 3$, might be non-trivial if $m>1$. 

We illustrate this phenomenon in the case $d =4$. An iterative procedure analogous to the construction in Subsection \ref{Sec:3_dimensional}, based again only on Lemma \ref{Lem:Macro1} and \ref{Lem:Macro2}, yields a continuous Bloch frame 
$\widehat \Phi : \partial \Bred^{(4)} \to \Hi^m$ satisfying all the relevant symmetries (on vertices, edges, faces and $3$-dimensional hyperfaces). By comparison with the input frame $\Psi : \Bred^{(4)}  \to \Hi^m$, one obtains a continuous map 
 $\widehat U : \partial \Bred^{(4)} \to \U(\C^m)$ such that  $\widehat \Phi(k) =  \Psi (k) \act  \widehat U(k)$ for all $k \in  \partial \Bred^{(4)} $.  Arguing as in Subsection \ref{Sec:3_dimensional}, one concludes that the existence of a continuous extension $\cell{U} : \Bred^{(4)} \to \U(\C^m)$ of $\widehat U$ is equivalent to the  fact that the homotopy class $[\widehat U]$ is the trivial element of  $\pi_3 \( \U(\C^m)\)$. Since the latter group is not trivial (for $m>1$), there might be \emph{a priori} topological obstruction to the existence of a continuous extension. This possible obstruction corresponds, in the abstract approach used in \cite{Panati2007, MoPa}, to the appearance  for $d \geq 4$ of the second Chern class of the Bloch bundle, which always vanishes for $d \leq 3$ or $m=1$. 
 
 On the other hand, our constructive algorithm works without obstruction in the case $m=1$, since $\pi_k \( \U(\C^1)\) = 0$ for 
 all $k \geq 2$, yielding an explicit construction of a global smooth symmetric Bloch frame.  However, since a constructive proof in the case $m=1$ is already known for every $d \in \N$ \cite{HeSj89}, we do not provide the details of the construction.



\section{A symmetry-preserving smoothing procedure}
\label{Sec:SmoothingProcedure}

In this Section we develop a smoothing procedure which, given a global continuous symmetric Bloch frame, 
yields a global \emph{smooth symmetric} Bloch frame arbitrarily close to the given one. The following Proposition, which holds true  in any dimension, might be of independent interest.  

\begin{prop}[Symmetry-preserving smoothing procedure] \label{Lem:SmoothingProcedure}
For $d \in \N$,  let $\mathcal{P} = \set{P(k)}_{k \in \R^d}$ be a family of orthogonal projectors
satisfying Assumption \ref{Ass:projectors}. Assume that $\Phi : \R^d \to \Hi^m$ is a global \emph{continuous 
symmetric}  Bloch frame, \ie it satisfies properties $(\mathrm{F}_{0})$, $(\mathrm{F}_{2})$ and 
$(\mathrm{F}_{3})$. 

\noindent Choose $\epsi > 0$. Then one constructs a global \emph{smooth symmetric}  Bloch frame 
$\Phi\sub{sm}$ such that 
\begin{equation} \label{Eq:Frame_distance}
\sup_{k \in \R^d} \dist \( \Phi(k),  \Phi\sub{sm}(k) \) < \epsi 
\end{equation} 
where  $\dist \( \Phi,  \Psi \) = \( \sum_a \norm{\phi_a - \psi_a}_{\Hi}^2 \)^{1/2}$ is the distance in $\Fr(m, \Hi)$.
\end{prop}

Notice that, for any  $\Phi \in \Fr(m, \Hi)$ and $U, W \in \U(m)$, one has 
\begin{equation} \label{Eq:U-W_distance}
\dist(\Phi \act U, \Phi \act W) =  \norm{U -W}\sub{HS}  
\end{equation}
where $\norm{U}\sub{HS}^2 = \sum_{a,b=1}^m | U_{ab}|^2$ is the Hilbert-Schmidt norm. Thus,  the distance between the frames 
$\Phi \act U$  and $\Phi \act W$   is the length of the chord between $U$ and $W$  in the ambient space $\C^{m^2} \simeq M_m(\C) \supset \U(m) $.   On the other hand,  each frame space $F_k := \Fr \Ran P(k) \simeq \U(m)$  inherits from $\U(m)$ a Riemannian structure
\footnote{ Recall that $\U(m)$ is a Riemannian manifold with respect to the bi-invariant metric defined, for $A, B$ in the Lie algebra $\mathfrak{u}(m) = \set{A \in M_m(\C) : A^*=-A}$, by  $\inner{A}{B}\sub{HS} = \tr (A^* B)$.}, 
and the corresponding geodesic distance $\di(U,W)$ can be compared to the chord distance \eqref{Eq:U-W_distance}.  In a neighborhood  of the identity, the  geodesic distance and the ambient distance are locally Lipschitz equivalent, namely
\begin{equation} \label{Eq:distance_equivalnce}
\half \, \di(\id, U) \leq  \norm{\id - U}\sub{HS}  \leq   \di(\id,U)      \qquad \qquad   \forall U \in \U(M) :\  \norm{\id - U}\sub{HS} < \half \tau_m,   
\end{equation}  
where $\tau_m$ is defined as the largest number having the following property: The open normal bundle over $\U(m)$ of radius $r$ is embedded in $\R^{2m^2} \simeq M_m(\C)$ for every $r < \tau_m$.  The first inequality in \eqref{Eq:distance_equivalnce} is a straightforward consequence of  \cite[Prop. 6.3]{NSW}, where  also the relation between $\tau_m$ and the principal curvature of $\U(m)$ is discussed.    

\begin{proof}[Proof] Following \cite{Panati2007},  we recall that, to a family of projectors $\mathcal{P}$ satisfying properties \ref{item:smooth} and \ref{item:tau}, one can canonically associate a {smooth} Hermitian vector bundle 
$\mathcal{E_P} = (E_{\mathcal{P}} \to \T^d_*)$,  where $\T^d_* = \R^d / \Lambda$.  In particular, $\mathcal{E_P}$ is defined by using an equivalence relation $\sim_{\tau}$ on $\R^d \times \Hi$, namely 
$$
(k, \phi) \sim_{\tau} (k', \phi')  \qquad  \text{ if and only if }   \qquad  \exists  \la \in \La :  k' = k + \la, \ \phi' = \tau_\la \phi . 
$$
An equivalence class is denoted by $[k, \phi]_\tau$. Then, the total space is defined by 
$$
E_{\mathcal{P}} = \set{[k, \phi]_\tau \in \( \R^d \times \Hi \)/\sim_{\tau}  : 
\phi \in \Ran P(k)},
$$
and the projection $\pi:  E_{\mathcal{P}} \to \T^d_*$ by $\pi([k, \phi]_\tau) = k $ (mod $\La$). The fact that  $\mathcal{E_P}$ is a smooth vector bundle follows from \ref{item:smooth} and the Kato-Nagy formula, see \cite{Panati2007, PanatiPisante} for the proof. Moreover, a natural Hermitian structure is induced by the inner product in $\Hi$. 

Equipped with the above definition, we observe that a  {continuous $\tau$-equivariant} global Bloch frame $\Phi :  \R^d \to \Hi^m$ is identified with a  {continuous} global section $\sigma_{\Phi} $ of the (principal) bundle of the orthonormal frames of the bundle $\mathcal{E_P}$, denoted by $\Fr \mathcal{E_P}$. The identification is given by 
$$
\sigma_{\Phi}(x) = \( [k, \phi_1(k)]_{\tau}, \ldots, [k, \phi_m(k)]_{\tau} \) \in \( \Fr E_{\PB}\)_x   \qquad \qquad \text{for }  x = k \text{ mod } \La.
$$

\noindent According to a classical result, the Steenrod's Approximation Theorem (\cite{Steenrod}; see \cite{Wockel} for recent generalizations), there exists a  {smooth} global section $\sigma'_{\Phi}  :   \T^d_* \to \Fr E_{\PB}$ such that
$$
\sup_{x \in\T^d_*} \dist \(\sigma_{\Phi}(x),  \sigma'_{\Phi}(x) \) < \half \, \epsi. 
$$

Going back to the language of Bloch frames, one concludes the existence of a global  {smooth $\tau$-equivariant} Bloch frame $\Phi' \sub{sm}: \R^d \to \Hi^m$, such that 
\begin{equation} \label{Eq:Frame_closeness}
\sup_{k \in \R^d} \dist \( \Phi(k), \Phi' \sub{sm}(k) \) < \half \, \epsi. 
\end{equation} 


In general,  the Bloch frame $\Phi' \sub{sm}$ does not satisfy property $(\mathrm{F}_{3})$.  In order to recover time-reversal symmetry, we use  the following symmetrization procedure. 


First, we recall that there exists $\delta >0$ such that the exponential map $\exp : \mathfrak{u}(m) \to \U(m)$ is a diffeomorphism from the ball $B_{\delta}(0) \subset  \mathfrak{u}(m)$ to the geodesic ball $B_{\delta}(\id) \subset \U(m)$, see \eg 
\cite[Chapter II]{Helgason} or \cite[Chapter VII]{Simon}. In particular, for any $U \in B_{\delta}(\id)$, there exists a unique $A_U \in B_{\delta}(0)$ such that
 \begin{equation} \label{Eq: exponential map}
U = \exp (A_U), \qquad \qquad U \in B_{\delta}(\id),   
\end{equation}
and, moreover, the map $U \mapsto A_U$ is $C^{\infty}$-smooth on $B_{\delta}(\id)$. Since the exponential map is normalized so that $\di(\id, U) = \norm{A_U}\sub{HS}$, then $\di(\id, \overline{U}) = \di(\id, U) = \di(\id, U^{-1})$. 
In particular, both $\overline{U}$ and $U^{-1}$ are in the geodesic ball $B_{\delta}(\id)$, whenever $U \in B_{\delta}(\id)$. 

\noindent For $U \in B_{\delta}(\id)$, the  {midpoint $M(\id, U)$} between $\id$ and $U$  is defined by
\footnote{Definition \eqref{Eq:geodesic_midpoint} agrees with the geodesic midpoint  between $\id$ and $U$ in the Riemannian manifold $\U(m)$, since the exponential map is normalized so that $\di(\id, \exp(sN)) = s$, for $s < \delta$ and $\norm{N}\sub{HS}=1$. }
\begin{equation} \label{Eq:geodesic_midpoint}
M(\id, U) := \exp (\half A_U) \qquad  \in B_{\delta}(\id) \subset \U(m).
\end{equation}
One immediately checks that, for $U \in B_{\delta}(\id)$, 
\begin{eqnarray}\label{Eq:M(I,U)} 
M(\id, \overline{U}) &=& \exp\( \half \overline{A_U} \) = \overline{M(\id, U)} \\
M(\id, U^{-1})  &=&  \exp\( - \half A_U \)  = U^{-1} M(\id, U).  \label{Eq:M(I,U*)}
\end{eqnarray}
Moreover, 
\begin{equation} \label{Eq:mid_distance}
\di(\id, M(\id, U)) = \half \di(\id, U). 
\end{equation}


\bigskip

Consider now two orthonormal frames $\Phi, \Psi \in \Fr \Ran P(k)$, such that $\dist(\Phi, \Psi) < \epsi$.  For $\epsi$ sufficiently 
small, we define the {\bf midpoint} $\midpoint{\Phi}{\Psi} \in \Fr \Ran P(k)$ in the following way.  

\noindent Let $U_{\Phi, \Psi} \in \U(m)$ be the unique unitary such that $\Psi = \Phi \act \, U_{\Phi, \Psi}$, namely $\(U_{\Phi, \Psi}\)_{ab} = \inner{\phi_a}{\psi_b}$.  Taking \eqref{Eq:U-W_distance} and \eqref{Eq:distance_equivalnce} into account, one has 
$$
\epsi > \dist(\Phi, \Psi) =  \dist(\Phi, \Phi \act \UP) = \norm{\id - \UP}\sub{HS} \geq  \half \,\, \di(\id, \UP).
$$

\noindent Choose $\epsi$ sufficiently small, namely $\epsi < \delta/2$. Then $\UP$ is in the geodesic ball $B_{\delta}(\id) \subset \U(m)$. By using  \eqref{Eq:geodesic_midpoint}, we define 
\begin{equation} \label{Eq:Fr_midpoint}
\midpoint{\Phi}{\Psi} :=  \Phi \act M(\id, \UP)            \qquad \in \Fr \Ran P(k).
\end{equation}


We  show that 
\begin{eqnarray} 
\midpoint{\Theta \, \Phi}{\Theta \Psi} =\Theta \, \midpoint{\Phi}{\Psi}
\label{Eq:midpoint_invariance_1}\\
\label{Eq:midpoint_invariance_2}
\midpoint{\tau_\la \Phi}{\tau_\la \Psi} =\tau_\la  \, \midpoint{\Phi}{\Psi}.
\end{eqnarray}

\noindent  Notice preliminarily that, since both $\Theta$ and $\tau_\la$ are isometries of $\Hi$, one has 
\begin{equation}  \label{Eq:dist_invariance}
\dist(\Theta \, \Phi, \Theta \Psi)  = \dist(\Phi, \Psi) = \dist(\tau_\la \Phi, \tau_\la \Psi)     \\
\end{equation}
for all  $\Phi, \Psi \in \Fr(m, \Hi)$. Thus, the midpoints appearing on the left-hand sides of  \eqref{Eq:midpoint_invariance_1}
and  \eqref{Eq:midpoint_invariance_2} are well-defined, whenever $\dist(\Phi, \Psi) < \half \delta$. 

\noindent Equation \eqref{Eq:midpoint_invariance_1} follows from  \eqref{Eq:M(I,U)}  and from
the fact that $\Theta \(\Phi \act \UP \) = \( \Theta \, \Phi\) \act \overline{\UP}$.  Indeed, one has
\begin{eqnarray*} 
\midpoint{\Theta \, \Phi}{\Theta \Psi} &=&  \midpoint{\Theta \, \Phi}{\Theta \(\Phi \act \UP \) } = \midpoint{\Theta \, \Phi}{\(\Theta \, \Phi\) \act \overline{\UP}} \\
&=&  \(\Theta \, \Phi \) \act M(\id, \overline{\UP}) = \( \Theta  \Phi \) \act \overline{M(\id, \UP)}  \\
&=& \Theta  \( \Phi  \act M(\id, \UP) \) = \Theta \, \midpoint{\Phi}{\Psi}. 
\end{eqnarray*}
\noindent Analogously, equation \eqref{Eq:midpoint_invariance_2} follows from the fact that $\tau_\la \(\Phi \act \UP \) = \( \tau_\la \Phi\) \act \UP$.

\bigskip

We focus now on  the smooth $\tau$-equivariant Bloch frame $\Phi' \sub{sm} :  \R^d \to \Hi^m$ obtained, via Steenrod's theorem, from the continuous symmetric frame $\Phi$.  Since $\Ran P(k) = \Theta \Ran P(-k)$, one has that  $\Theta \,\Phi' \sub{sm}(-k)$ is in  $F_k = \Fr \Ran P(k)$.  Thus, we set
\begin{equation} \label{Def:Symmetrized_section}
\Phi \sub{sm}(k) := \midpoint{\Phi' \sub{sm}(k)}{\Theta \,\Phi' \sub{sm}(-k)}  \qquad \in F_k. 
\end{equation}  
The definition \eqref{Def:Symmetrized_section} is well-posed.  Indeed, taking \eqref{Eq:Frame_closeness} and 
\eqref{Eq:dist_invariance} into account, one has 
\begin{eqnarray}\label{Eq:Frame_comparison} 
&\hspace{-2cm}&\dist(\Phi' \sub{sm}(k), \Theta \,\Phi' \sub{sm}(-k)) \\
& \leq & \dist(\Phi' \sub{sm}(k), \Phi(k)) + \dist(\Phi(k), \Theta \, \Phi(-k))  +  \dist(\Theta \, \Phi(-k), \Theta \,\Phi' \sub{sm}(-k))  
\nonumber \\
& = &   \dist(\Phi' \sub{sm}(k), \Phi(k))  +  \dist(\Phi(-k), \Phi' \sub{sm}(-k))  <  \epsi < \half \delta,  
\nonumber
\end{eqnarray}
where we used the fact that the central addendum (in the second line) vanishes since $\Phi$ satisfies $(\mathrm{F}_{3})$.  

We claim that \eqref{Def:Symmetrized_section} defines a smooth symmetric global Bloch frame satisfying  \eqref{Eq:Frame_distance}. We explicitly check that:

\textbf{1.} the map $k \mapsto \Phi \sub{sm}(k)$ is smooth. Indeed, since $\Theta$ is an isometry of $\Hi$, the map
$k \mapsto \Theta \,\Phi' \sub{sm}(-k)=: \Psi'\sub{sm}(k)$ is smooth. Hence $k \mapsto U_{{\Phi' \sub{sm}}(k),{\Psi' \sub{sm}}(k)} \in \U(m)$ is smooth,  since  $ \(U_{\Phi, \Psi}\)_{ab} = \inner{\phi_a}{\psi_b}$.  In view of  \eqref{Eq:Frame_comparison}  and \eqref{Eq:distance_equivalnce},  $U_{\Phi' \sub{sm}(k), \Psi'\sub{sm}(k)} $  is, for every $k \in \R^d$, in the geodesic ball $B_{\delta}(\id)$ where the exponential map defines a diffeomorphism.  As a consequence, 
 $$
 k \mapsto\Phi' \sub{sm}(k) \act M(\id, U_{\Phi' \sub{sm}(k), \Psi' \sub{sm}(k)}) =  \Phi\sub{sm}(k)  
 $$
is smooth from $\R^d$ to $\Hi^m$. 

\textbf{2.} the Bloch frame $\Phi \sub{sm}$ satisfies $(\mathrm{F}_{2})$. Indeed, by using 
$(\mathrm{P}_{4})$ and \eqref{Eq:midpoint_invariance_2}, one obtains
\begin{eqnarray*} 
\label{Eq:phi_tau_eq}
\Phi\sub{sm}(k + \la) &=&   \midpoint{\Phi'\sub{sm}(k + \la) }{ \Theta \,\Phi' \sub{sm}(- k- \la)} \\
 &=&   \midpoint{\tau_\la \,\Phi' \sub{sm}(k) }{ \Theta \, \tau_{- \la}  \,\Phi' \sub{sm}(- k)} \\
 &=&   \midpoint{\tau_\la \,\Phi' \sub{sm}(k) }{ \tau_{\la} \, \Theta  \,\Phi' \sub{sm}(- k)}  \\
 &=&   \tau_\la \, \midpoint{\Phi'\sub{sm}(k) }{\Theta  \,\Phi' \sub{sm}(- k)} =  \tau_\la \, \Phi\sub{sm}(k). 
\end{eqnarray*}

\textbf{3.} the Bloch frame $\Phi \sub{sm}$ satisfies $(\mathrm{F}_{3})$. Indeed, by using 
$\Theta^2 = \Id$ and \eqref{Eq:midpoint_invariance_1}, one has  
\begin{eqnarray*} 
\label{Eq:phi_TR}
\Phi\sub{sm}(- k) &=&   \midpoint{ \Theta^2 \,\Phi' \sub{sm}(- k) }{ \Theta \,\Phi' \sub{sm}(k)} \\
 &=&   \Theta \, \midpoint{\Theta \,\Phi' \sub{sm}(- k) }{\Phi' \sub{sm}(k)} \\
 &=&  \Theta \, \midpoint{\Phi'\sub{sm}(k) }{ \Theta \,\Phi' \sub{sm}(- k)} =   \Theta \, \Phi\sub{sm}(k),
 \end{eqnarray*}
where we used the fact that $\midpoint{\Phi}{\Psi}=\midpoint{\Psi}{\Phi}$, whenever $\dist(\Phi, \Psi) < \delta/2$. The latter fact 
is a direct consequence of  \eqref{Eq:M(I,U*)}, since 
\begin{eqnarray*} 
\midpoint{\Phi}{\Psi} &=& \Phi \act M(\id,\UP) = \(\Psi \act \UP^{-1}\) \act  M(\id,\UP) \\
&=&  \Psi \act \( \UP^{-1}  \, M(\id,\UP) \)  =  \Psi \act M(\id,\UP^{-1}) \\ 
&=&  \Psi \act M(\id, U_{\Psi, \Phi}) =  \midpoint{\Psi}{\Phi}.
\end{eqnarray*}

\textbf{4.}  equation  \eqref{Eq:Frame_distance} is satisfied in view of \eqref{Eq:Frame_comparison}. Indeed, 
setting $U_{{\Phi'\sub{sm}}(k),{\Psi'\sub{sm}}(k)} \equiv U(k)$ for notational convenience and using \eqref{Eq:U-W_distance}, \eqref{Eq:distance_equivalnce} and \eqref{Eq:mid_distance}, one obtains 
\begin{eqnarray*} \label{Eq:small_distance}
\dist(\Phi'\sub{sm}(k), \Phi\sub{sm}(k)) &=&  \dist \(\Phi' \sub{sm}(k),\Phi' \sub{sm}(k) \act M(\id, U(k)) \) \\
&=& \norm{\id - M(\id, U(k))}\sub{HS}  \leq  \di\(\id , M(\id, U(k))\) = \half \, \di( \id , U(k)) \\
&\leq& \norm{\id - U(k)}\sub{HS} = \dist\(\Phi'\sub{sm}(k), \Theta \,\Phi' \sub{sm}(- k)\) < \epsi.  
\end{eqnarray*}


This concludes the proof of the Proposition.  
\end{proof}


\bigskip \bigskip


{\footnotesize  

\begin{tabular}{ll}
(D. Fiorenza) & \textsc{Dipartimento di Matematica,   \virg{La Sapienza} Universit\`{a} di Roma} \\
 &  Piazzale Aldo Moro 2, 00185 Rome, Italy \\
 &  {E-mail address}: \href{mailto:fiorenza@mat.uniroma1.it}{\texttt{fiorenza@mat.uniroma1.it}} \\
 \\
(D. Monaco) &  \textsc{SISSA  International School for Advanced Studies}\\
&Via Bonomea 265, 34136 Trieste, Italy \\
&{E-mail address}: \href{mailto:dmonaco@sissa.it}{\texttt{dmonaco@sissa.it}} \\
\\
(G. Panati) & \textsc{Dipartimento di Matematica, \virg{La Sapienza} Universit\`{a} di Roma} \\
 &  Piazzale Aldo Moro 2, 00185 Rome, Italy \\
 &  {E-mail addresses}: \href{mailto:panati@mat.uniroma1.it}{\texttt{panati@mat.uniroma1.it}}, 
 \href{mailto:panati@sissa.it}{\texttt{panati@sissa.it}}  \\

\end{tabular}

\bigskip

}
\end{document}